\documentclass[twocolumn,superscriptaddress,aps]{revtex4-1}
\usepackage{graphicx}
\usepackage{amsmath}
\usepackage{amsthm}
\usepackage{amssymb}
\usepackage{latexsym}
\usepackage{array}
\usepackage{hyperref}
\usepackage{float}
\usepackage{amsfonts}
\usepackage{dsfont}
\usepackage{mathrsfs}
\usepackage{verbatim}
\usepackage{bbold}
\usepackage{lipsum}
\usepackage[normalem]{ulem}

\usepackage{upgreek}

\usepackage{makecell}
\usepackage{adjustbox,lipsum}

\usepackage{algorithm}
\usepackage{algpseudocode}

\usepackage{color}

\newcommand{\bra}[1]{\left<#1\right|}
\newcommand{\ket}[1]{\left|#1\right>}
\newcommand{\abs}[1]{\left|#1\right|}
\newcommand{\norm}[1]{\left\lVert#1\right\rVert}

\newcommand{\braket}[2]{\left<{#1}|{#2}\right>}
\newcommand{\ketbra}[2]{\ket{#1}\!\!\bra{#2}}

\newtheorem{theorem}{Theorem}
\newtheorem{proposition}{Proposition}
\newtheorem{lemma}{Lemma}

\newtheorem{definition}{Definition}

\newtheorem{assumption}{Assumption}
\newtheorem{remark}{Remark}

\DeclareMathOperator{\Tr}{Tr}

\begin{document}

\title{Run-length certificates in quantum learning: sample complexity and noise thresholds}

\author{Jeongho~Bang}\email{jbang@yonsei.ac.kr}
\affiliation{Institute for Convergence Research and Education in Advanced Technology, Yonsei University, Seoul 03722, Republic of Korea}
\affiliation{Department of Quantum Information, Yonsei University, Incheon 21983, Republic of Korea}

\date{\today}% It is always \today, today, but any date may be explicitly specified

\begin{abstract}
Quantum learning from state samples is often benchmarked in a fixed-budget paradigm, relating error to a prescribed number of copies. We instead adopt a stopping-time viewpoint: in minimal-feedback learning, the learning completion can be defined by an online run-length certificate extracted from a one-bit-per-copy record. As an instantiation, we analyze single-shot measurement learning (SSML), introduced in [Phys.~Rev.~A {\bf 98},~052302~(2018)] and [Phys.~Rev.~Lett.~{\bf 126},~170504~(2021)], which tunes a unitary and halts after $M_H$ consecutive successes. Viewing the halting as a sequential certification linking the observed counter to infidelity, we derive sample-complexity bounds that separate \emph{search} (driving success probability toward unity) from \emph{certification} (run statistics of consecutive successes). The resulting trade-off among $M_H$, dimension $d$, and one-bit reliability clarifies when performance is control-limited versus certificate-limited. With label-flip noise probability $q$, we find a sharp feasibility threshold: once $qM_H \gtrsim 1$, the expected halting time grows exponentially, making the learning completion impractical even under ideal control. More broadly, this shows that under severely constrained feedback, the certification can dominate sample complexity and small label noise becomes the information bottleneck. Finally, the near-optimal accuracy enabled by run-length certification aligns with the quantum-state-estimation (and equivalently, no-cloning) limits, expressed in the stopping-time terms.
\end{abstract}

\maketitle

%-------------------------------------------------------------------------------------------------------------------------------------------------------------------------------------------------------------------------------------
\section{Introduction}\label{sec:intro}
%-------------------------------------------------------------------------------------------------------------------------------------------------------------------------------------------------------------------------------------

State learning from finitely many copies is a canonical arena where the operational constraints of laboratory measurements meet fundamental limits of quantum information~\cite{Holevo2011,Zhao2024,Anshu2024,Bang2026}. In the standard ``fixed-budget'' paradigm, one specifies an experimental sample budget $N$ in advance and asks how the estimation error scales with $N$~\cite{James2001,Teo2015}. Yet, many modern platforms operate naturally in a closed-loop mode: the controller adapts measurements in real time, updates the control based on minimal classical feedback, and decides on the fly when the acquired evidence is sufficient to stop. In such settings, the central resource is no longer a predetermined $N$ but the random number of copies consumed until a data-dependent termination event occurs.

Single-shot measurement learning (SSML) is an especially transparent representative of this adaptive paradigm~\cite{Lee2018, Lee2021}. Introduced as a minimal-feedback learning algorithm for unknown pure states, SSML repeatedly (i) applies a tunable unitary control to each fresh copy of the unknown state and (ii) performs a binary test against a fixed fiducial state, producing only a single ``success or failure'' bit per copy. The controller updates the unitary by a simple random-perturbation rule driven solely by this one-bit record and a tiny internal memory (a counter of consecutive successes). SSML thus isolates an extreme regime of information acquisition: \emph{one quantum copy $\to$ one classical bit $\to$ one causal update}.

What makes SSML conceptually distinctive is that it defines ``learning completion'' intrinsically through a halting rule. Rather than terminating after a fixed number of copies, the algorithm stops once it observes $M_H$ consecutive recorded successes. Because the SSML control is frozen during success events, the terminal run of $M_H$ successes is generated under a fixed control setting. This imbues the halting event with a concrete statistical meaning: it acts as an online sequential certificate that the success probability---and hence the state fidelity to the target---is close to unity. In other words, SSML intertwines the learning and certification in a single experimentally implementable loop, and the sample usage itself becomes an object of inference.

Despite extensive numerical and experimental studies, SSML is most often discussed in the language of fixed-budget accuracy (infidelity versus $N$), while the halting rule is treated largely heuristically. The guiding question of this work is therefore not only ``how accurate is SSML after $N$ shots?'' but, more operationally, how many copies are required until SSML halts, as a function of the certificate threshold $M_H$, the Hilbert-space dimension $d$, and the reliability of the one-bit feedback record? Addressing this question demands a stopping-time viewpoint: the halting time is a random variable whose expectation, tail behavior, and high-confidence quantiles constitute the natural notion of sample complexity.

Our analysis proceeds by reformulating SSML as a general stopping-time problem and by separating two qualitatively different sources of sample cost. First, there is a search cost: the dimension-dependent effort needed to steer the control into a region where the success probability is sufficiently close to one. Second, there is a certification cost: once in such a region, the additional overhead of accumulating $M_H$ consecutive successes, which is governed by universal run statistics and is largely independent of $d$. This decomposition clarifies which aspects of SSML performance are algorithmic (control landscape and update design) and which are statistical (evidence accumulation under a one-bit record).

A second theme of this paper is robustness. Because SSML relies on each single-bit label to update and to certify completion, the classification noise---flips of the recorded success/failure label~\cite{Angluin1988,Cover1999}---can have an outsize impact on the stopping time even when perfect control is, in principle, available. We show that the halting rule induces a sharp operational feasibility condition: in the regime where the expected number of corrupted labels within a would-be certificate run is $O(1)$, i.e., $qM_H \gtrsim 1$ for flip probability $q$, the expected halting time necessarily exhibits an exponential blow-up. This phenomenon is certification-limited and persists even for an idealized ``oracle'' search stage~\cite{Drekic2021}.

Finally, we interpret these stopping-time scaling laws through the lens of quantum-information principles. The near-optimal $O(N^{-1})$ SSML accuracy reported in Ref.~\cite{Lee2021} is consistent with the information limits of (pure) state learning and, equivalently, with no-cloning constraints. Our contribution is to show how such limits manifest in a sequential language: the halting condition provides an operational bridge from a one-bit-per-copy record to a quantitative certificate, while the noise threshold illustrates how physical imperfections in the classical feedback channel translate into an information-theoretic bottleneck for adaptive learning.

%-------------------------------------------------------------------------------------------------------------------------------------------------------------------------------------------------------------------------------------
\section{SSML algorithm and notation}\label{sec:ssml_alg}
%-------------------------------------------------------------------------------------------------------------------------------------------------------------------------------------------------------------------------------------

First, we specify the single-shot measurement learning (SSML) and fix the notation used throughout the paper. Our presentation is deliberately operational: SSML is a closed-loop procedure that consumes one copy of an unknown state per step, produces a single binary outcome, and updates a controllable unitary in response. The algorithm terminates when a run-length halting condition is satisfied, and this termination event is itself the natural notion of ``learning completion.''

%---------------------------------------------------------------------------------------------------------------------------------------------------------------------------------
\subsection{Physical setting and learning objective}\label{subsec:ssml_setting}

Let $\mathcal{H} \simeq \mathbb{C}^{d}$ be a finite-dimensional Hilbert space.  A preparation device $\mathsf{P}$ repeatedly emits an unknown pure state
\begin{eqnarray}
\hat{\rho}_{\tau} := \ketbra{\psi_{\tau}}{\psi_{\tau}}, \quad \ket{\psi_{\tau}}\in\mathcal{H},
\label{eq:ssml_unknown_state}
\end{eqnarray}
and the only promised side information is the Hilbert-space dimension $d$. In SSML one does not attempt to ``open'' the device $\mathsf{P}$; instead, one infers and reproduces $\ket{\psi_{\tau}}$ solely through sequential interactions with its outputs.

The learner controls a family of unitary operations
\begin{eqnarray}
\hat{U}(\mathbf{p})\in\mathsf{U}(d), \quad \mathbf{p}\in\mathcal{P}\subset\mathbb{R}^{m},
\label{eq:ssml_unitary_family}
\end{eqnarray}
where $\mathbf{p}$ is a classical control vector updated during learning. A fiducial (reference) state $\ket{f} \in \mathcal{H}$ is fixed in advance and chosen to be experimentally convenient (e.g., the most reliably preparable and detectable state in a given platform). The learning objective is to find a parameter $\mathbf{p}$ such that $\hat{U}(\mathbf{p})$ maps the unknown state to the fiducial state $\ket{f}$ up to a global phase,
\begin{eqnarray}
\hat{U}(\mathbf{p}_{\star})\ket{\psi_{\tau}} = e^{i\theta_{\star}}\ket{f},
\label{eq:ssml_objective}
\end{eqnarray}
so that the unknown input state can be inferred (and reproduced) as
\begin{eqnarray}
\ket{\psi_{\tau}} \simeq \ket{\psi_{\tau,{\rm est}}} := \hat{U}(\mathbf{p}_{\rm est})^{\dagger}\ket{f},
\label{eq:ssml_state_estimate}
\end{eqnarray}
where $\mathbf{p}_{\rm est} \simeq \mathbf{p}_{\star}$. Thus, SSML runs based on the following assumption:
\begin{assumption}[Reachability]\label{assump:ssml_reachability}
There exists $\mathbf{p}_{\star}\in\mathcal{P}$ such that Eq.~(\ref{eq:ssml_objective}) holds.
\end{assumption}

Eq.~(\ref{eq:ssml_state_estimate}) expresses a practical advantage of the unitary-learning approach: once $\hat{U}(\mathbf{p}_{\rm est})$ is learned, reproduction is realized directly by applying $\hat{U}(\mathbf{p}_{\rm est})^{\dagger}$ to a known input $\ket{f}$, without requiring a separate tomographic reconstruction of $\ket{\psi_{\tau}}$~\cite{Lee2018,Lee2021}.

%---------------------------------------------------------------------------------------------------------------------------------------------------------------------------------
\subsection{Single-shot measurement and success probability}\label{subsec:ssml_measurement}

At each learning step, SSML performs a two-outcome projective measurement that asks a single ``yes-or-no'' question: Is the (processed) state equal to the fiducial state? Thus, we define
\begin{eqnarray}
\hat{M}_{f}:=\ketbra{f}{f}, \quad \hat{M}_{f^{\perp}}:=\mathds{1}-\hat{M}_{f}.
\label{eq:ssml_meas_def}
\end{eqnarray}
Given a current control parameter $\mathbf{p}$, one copy of $\hat{\rho}_{\tau}$ is transformed by $\hat{U}(\mathbf{p})$ and then measured by $\{\hat{M}_{f},\hat{M}_{f^{\perp}}\}$.  We label the outcome corresponding to $\hat{M}_{f}$ as a success ($s$) and the other outcome as a failure ($u$). The resulting success probability is
\begin{eqnarray}
F(\mathbf{p}) := \Tr\Big[\hat{M}_{f}\hat{U}(\mathbf{p})\hat{\rho}_{\tau}\hat{U}(\mathbf{p})^{\dagger}\Big] = \abs{\bra{f}\hat{U}(\mathbf{p})\ket{\psi_{\tau}}}^{2},
\label{eq:ssml_F_eps}
\end{eqnarray}
and we define the associated infidelity by
\begin{eqnarray}
\epsilon(\mathbf{p}) := 1-F(\mathbf{p}).
\label{eq:ssml_infidelity}
\end{eqnarray}
Since each iteration produces only a single bit $m_{n} \in \{s,u\}$, SSML is naturally viewed as an adaptive information-acquisition process operating under a one-bit feedback channel from the quantum system to the classical controller.

%---------------------------------------------------------------------------------------------------------------------------------------------------------------------------------
\subsection{Update rule, internal memory, and halting condition}\label{subsec:ssml_update}

The distinctive feature of SSML is that the update decision uses only the current single-shot outcome and a minimal internal memory recording the ``run length'' of consecutive successes. Let $M_{S}^{(n)}\in\mathbb{N}_{0}$ denote the number of consecutive success outcomes observed up to step $n$. Starting from $M_{S}^{(0)}=0$, the memory variable is updated as
\begin{eqnarray}
M_{S}^{(n)} = \left\{
\begin{array}{ll}
M_{S}^{(n-1)}+1, & \text{if}~m_{n}=s,\\
0, & \text{if}~m_{n}=u.
\end{array}
\right.
\label{eq:ssml_memory_update}
\end{eqnarray}
The control parameter $\mathbf{p}^{(n)}$ is updated only upon failures:
\begin{eqnarray}
\mathbf{p}^{(n+1)}=
\left\{
\begin{array}{ll}
\mathbf{p}^{(n)}, & \text{if}~m_{n}=s,\\
\mathbf{p}^{(n)}+\omega_{n}\mathbf{r}_{n}, & \text{if}~m_{n}=u,
\end{array}
\right.
\label{eq:ssml_param_update}
\end{eqnarray}
where $\mathbf{r}_{n}$ is a random direction and $\omega_{n}$ is a weighted step size of the form
\begin{eqnarray}
\omega_{n} := \alpha\big(M_{S}^{(n-1)}+1\big)^{-\beta}, \quad (\alpha>0,\ \beta>0).
\label{eq:ssml_weight_def}
\end{eqnarray}
The interpretation is intuitive: if many consecutive successes occur before a failure, the algorithm treats the current control as ``near optimal'' and thus makes a smaller random move; frequent failures yield larger exploratory steps. The completion of the learning is declared when the consecutive-success counter reaches a prescribed threshold $M_H$:
\begin{eqnarray}
M_{S}^{(n)}=M_{H}.
\label{eq:ssml_halting}
\end{eqnarray}
We refer to Eq.~(\ref{eq:ssml_halting}) as the ``halting condition,'' and we emphasize that it defines completion through an observable event (a long run of successes) rather than through an externally imposed resource budget.

\begin{remark}[Freezing on success and the certification viewpoint]\label{rem:ssml_freezing}
A structural feature that will play an important role later is that successes \emph{freeze} the parameter, i.e., $\mathbf{p}^{(n+1)}=\mathbf{p}^{(n)}$ when $m_{n}=s$. Hence, the final run of $M_{H}$ successes is generated under a fixed control $\mathbf{p}_{\rm est}$. This makes the halting condition naturally interpretable as an intrinsic sequential certificate for $F(\mathbf{p}_{\rm est})$, and hence for $\epsilon(\mathbf{p}_{\rm est})$, rather than as a purely algorithmic stopping rule.
\end{remark}

For completeness and to fix notation, {\bf Algorithm~\ref{alg:ssml_protocol}} summarizes the SSML loop in the noiseless setting.
\begin{algorithm}[H]
\caption{Single-shot measurement learning (SSML)}
\label{alg:ssml_protocol}
\begin{algorithmic}[1]
\State \textbf{Input:} fiducial state $\ket{f}$; halting threshold $M_H$; hyperparameters $\alpha,\beta$
\State Initialize $\mathbf{p}^{(1)}\leftarrow \mathbf{r}$ at random; set $M_{S}^{(0)}\leftarrow 0$
\For{$n=1,2,\dots$}
    \State Prepare one copy of $\hat{\rho}_{\tau}$ and apply $\hat{U}(\mathbf{p}^{(n)})$
    \State Measure $\{\hat{M}_{f}, \hat{M}_{f^{\perp}}\}$ and record $m_{n}\in\{s,u\}$
    \If{$m_{n}=s$}
        \State $M_{S}^{(n)}\leftarrow M_{S}^{(n-1)}+1$; \quad $\mathbf{p}^{(n+1)}\leftarrow \mathbf{p}^{(n)}$
    \Else
        \State $M_{S}^{(n)}\leftarrow 0$
        \State $\omega_{n}\leftarrow \alpha\big(M_{S}^{(n-1)}+1\big)^{-\beta}$
        \State Sample $\mathbf{r}_{n}$ and set $\mathbf{p}^{(n+1)}\leftarrow \mathbf{p}^{(n)}+\omega_{n}\mathbf{r}_{n}$
    \EndIf
    \If{$M_{S}^{(n)}=M_H$}
        \State \textbf{halt} and output $\mathbf{p}_{\rm est}\leftarrow \mathbf{p}^{(n)}$
    \EndIf
\EndFor
\end{algorithmic}
\end{algorithm}

%---------------------------------------------------------------------------------------------------------------------------------------------------------------------------------
\subsection{Parameterization and dimension dependence}\label{subsec:ssml_param}

The SSML update rule itself is agnostic to how $\hat{U}(\mathbf{p})$ is implemented. For analysis it is useful to keep track of the effective parameter dimension $m$ and its dependence on the Hilbert-space dimension $d$. A generic choice is to parameterize $\hat{U}(\mathbf{p})$ via $\mathfrak{su}(d)$ generators $\{\hat{G}_{j}\}_{j=1}^{d^{2}-1}$,
\begin{eqnarray}
\hat{U}(\mathbf{p}) = \exp\!\Big(-i\sum_{j=1}^{d^{2}-1} p_{j}\hat{G}_{j}\Big),
\quad
\mathbf{p}\in\mathbb{R}^{d^{2}-1},
\label{eq:ssml_sud_param}
\end{eqnarray}
so that $m=d^{2}-1$ in a fully general model. For qubits ($d=2$), one may equivalently write $\hat{U}(\mathbf{p})=\exp(-i\,\mathbf{p}^{T}\hat{\mathbf{G}})$ with $\hat{\mathbf{G}}=(\hat{\sigma}_{x},\hat{\sigma}_{y},\hat{\sigma}_{z})^{T}$.

While the performance of SSML depends on the choice of parametrization and on experimental constraints, our stopping-time analysis will emphasize two universal control parameters that appear already at the algorithm level: (i) the halting threshold $M_H$, which sets the strength of the run-length certificate, and (ii) the Hilbert-space dimension $d$, which controls both the intrinsic estimation difficulty and, in generic control models, the parameter dimension $m$.

%---------------------------------------------------------------------------------------------------------------------------------------------------------------------------------
\subsection{High-dimensional extension: multi-failure outcomes}\label{subsec:ssml_highd}

For completeness, we note a standard extension of SSML to $d>2$ in which the single ``failure'' outcome is refined into multiple failure labels. Fix an orthonormal basis $\{\ket{f},\ket{f_{1}^{\perp}},\dots,\ket{f_{d-1}^{\perp}}\}$ and consider the projective measurement
\begin{eqnarray}
\Big\{\ketbra{f}{f},\ \ketbra{f_{1}^{\perp}}{f_{1}^{\perp}},\ \dots,\ \ketbra{f_{d-1}^{\perp}}{f_{d-1}^{\perp}}\Big\}.
\label{eq:ssml_highd_meas}
\end{eqnarray}
A success event corresponds to detecting $\ket{f}$, while a failure event carries an additional label $\ell \in \{1,\dots,d-1\}$ specifying which orthogonal outcome occurred. Operationally, this label can be used to restrict the random update $\mathbf{r}_{n}$ to a subspace selected by the observed failure outcome (for example, the two-dimensional subspace spanned by $\ket{f}$ and $\ket{f_{\ell}^{\perp}}$), thereby reducing unnecessary motion in directions that do not contribute to increasing $F(\mathbf{p})$.  The stopping-time viewpoint adopted later applies equally to this multi-failure variant: the halting event remains a run of successes, while the search dynamics are modified through the update directions.

\begin{remark}[Scope]\label{rem:ssml_scope}
We will present our main results for the binary-outcome formulation $\{\hat{M}_{f},\hat{M}_{f^{\perp}}\}$ for notational simplicity. The extension to multi-failure outcomes primarily changes the geometry of the search dynamics and the $d$-dependence of the constants, while the certification component of the stopping-time analysis is governed by the same run-length statistics.
\end{remark}

%-------------------------------------------------------------------------------------------------------------------------------------------------------------------------------------------------------------------------------------
\section{Stopping-time formulation and performance criteria}\label{sec:stopping_time_formulation}
%-------------------------------------------------------------------------------------------------------------------------------------------------------------------------------------------------------------------------------------

SSML is not naturally described by a fixed-shot budget. The algorithm consumes one copy per step and terminates only when the run-length counter reaches the halting threshold $M_H$. Accordingly, the primary resource is a random variable: the number of copies used until completion. In this section, we formalize SSML as a general stopping-time problem and introduce performance measures that match its operational semantics.

%---------------------------------------------------------------------------------------------------------------------------------------------------------------------------------
\subsection{SSML as an adapted stochastic process}\label{subsec:stopping_process}

We denote by $\mathbf{p}^{(n)}$ the control parameter at step $n$, by $m_{n} \in \{s,u\}$ the recorded single-shot label, and by $M_{S}^{(n)}\in\mathbb{N}_{0}$ the consecutive-success counter defined in
Eq.~(\ref{eq:ssml_memory_update}). The pair $\big(\mathbf{p}^{(n)}, M_{S}^{(n)}\big)$ evolves stochastically due to the Born-rule randomness of the binary measurement and the internal randomness used to generate the perturbation directions $\mathbf{r}_{n}$ in Eq.~(\ref{eq:ssml_param_update}).

For probabilistic statements, we work with the natural filtration
\begin{eqnarray}
\mathcal{F}_{n} := \sigma\Big(m_{1},\dots,m_{n} \ ; \ \mathbf{r}_{1},\dots,\mathbf{r}_{n} \ ; \ \mathbf{p}^{(1)}\Big),
\label{eq:filtration_def}
\end{eqnarray}
i.e., the $\sigma$-algebra generated by all information revealed up to step $n$~\cite{Durrett2019}. By construction, the SSML dynamics is adapted to $\{\mathcal{F}_{n}\}_{n\ge 0}$.

Throughout, $F(\mathbf{p})$ and $\epsilon(\mathbf{p})=1-F(\mathbf{p})$ denote the success probability and infidelity, defined in Eqs.~(\ref{eq:ssml_F_eps}) and (\ref{eq:ssml_infidelity}). Conditioned on $\mathcal{F}_{n-1}$, the $n$th measurement label obeys
\begin{eqnarray}
\mathbb{P}\big(m_{n}=s \,\big|\, \mathcal{F}_{n-1}\big) &=& F\big(\mathbf{p}^{(n)}\big), \nonumber \\
\mathbb{P}\big(m_{n}=u \,\big|\, \mathcal{F}_{n-1}\big) &=& 1-F\big(\mathbf{p}^{(n)}\big).
\label{eq:conditional_outcome}
\end{eqnarray}

%---------------------------------------------------------------------------------------------------------------------------------------------------------------------------------
\subsection{Stopping time and sample-complexity metrics}\label{subsec:stopping_time_def}

The halting rule of SSML is the event that the consecutive-success counter reaches the threshold $M_H$ (see Eq.~(\ref{eq:ssml_halting})). This induces the stopping time
\begin{eqnarray}
T := \inf\big\{n \ge 1 \ : \ M_{S}^{(n)}=M_H\big\}.
\label{eq:stopping_time_T}
\end{eqnarray}
By definition, $\{T \le N\} \in \mathcal{F}_{N}$ for all $N$, hence $T$ is a stopping time with respect to $\{\mathcal{F}_{n}\}$.

A basic operational descriptor is the ``learning probability'' by time $N$~\cite{Awaya2024,Bang2025},
\begin{eqnarray}
P(N) :=\mathbb{P}(T\le N),
\label{eq:learning_prob_PN}
\end{eqnarray}
which is a nondecreasing function of $N$ and converges to $1$ when $T$ is almost surely finite.

\begin{definition}[Expected and high-confidence sample complexity]\label{def:sample_complexity}
When $\mathbb{E}[T] < \infty$, the expected sample usage is $\mathbb{E}[T]$. For a confidence parameter $\delta\in(0,1)$, the \emph{high-confidence} sample complexity is defined as the $(1-\delta)$-quantile
\begin{eqnarray}
T_{\delta} := \inf\big\{N\in\mathbb{N} \ : \ P(N)\ge 1-\delta\big\}.
\label{eq:quantile_Tdelta}
\end{eqnarray}
\end{definition}

The quantity $T_{\delta}$ is the natural metric when an experiment imposes a strict shot budget: it specifies how many copies are sufficient to complete SSML with probability at least $1-\delta$.

%---------------------------------------------------------------------------------------------------------------------------------------------------------------------------------
\subsection{Accuracy at halting and intrinsic certification}\label{subsec:accuracy_at_halting}

SSML outputs the control $\mathbf{p}_{\rm est} := \mathbf{p}^{(T)}$ and the corresponding state estimate $\hat{U}(\mathbf{p}_{\rm est})^{\dagger}\ket{f}$. The achieved infidelity at completion is therefore the random variable
\begin{eqnarray}
\epsilon_{T} := \epsilon\big(\mathbf{p}^{(T)}\big) = 1-F\big(\mathbf{p}^{(T)}\big).
\label{eq:eps_at_halting}
\end{eqnarray}
A stopping-time theory must control both the resource usage $T$ and the achieved accuracy $\epsilon_T$, and, crucially, relate them through the halting rule itself.

The key operational fact is that SSML freezes the control on success events ({\bf Remark~\ref{rem:ssml_freezing}}). Hence, the terminal run of $M_H$ successes is generated under a fixed control, and can be interpreted as an intrinsic sequential certificate.

\begin{lemma}[Certificate strength]\label{lem:certificate_strength}
Fix any $\mathbf{p} \in \mathcal{P}$ and suppose $M_H$ measurements are performed under the same control $\hat{U}(\mathbf{p})$ and the same binary test $\{\hat{M}_{f}, \hat{M}_{f^{\perp}}\}$. Then, the probability of observing $M_H$ consecutive successes is
\begin{eqnarray}
\mathbb{P}\big(\text{$M_H$ consecutive successes} \mid \mathbf{p}\big) = F(\mathbf{p})^{M_H}.
%= \big(1-\epsilon(\mathbf{p})\big)^{M_H}.
\label{eq:certificate_prob}
\end{eqnarray}
Consequently, for any $\epsilon_{0} \in (0,1)$,
\begin{widetext}
\begin{eqnarray}
\epsilon(\mathbf{p})\ge \epsilon_{0} \ \Longrightarrow \ \mathbb{P}\big(\text{$M_H$ consecutive successes} \mid \mathbf{p}\big) \le (1-\epsilon_{0})^{M_H}.
\label{eq:certificate_bound}
\end{eqnarray}
\end{widetext}
\end{lemma}

{\bf Lemma~\ref{lem:certificate_strength}} shows that the halting threshold $M_H$ directly tunes a statistical error exponent. For a desired significance level $\delta \in(0,1)$, choosing $\epsilon_{0}$ such that $(1-\epsilon_{0})^{M_H} \le \delta$ yields the certified accuracy scale, for $M_H \gg 1$,
\begin{eqnarray}
\epsilon_{\rm cert}(M_H, \delta) := 1 - \delta^{1/M_H} \approx \frac{\ln(1/\delta)}{M_H}.
\label{eq:eps_cert}
\end{eqnarray}
Operationally, $M_H$ is therefore an evidence budget: longer runs correspond to stronger online certification.

%---------------------------------------------------------------------------------------------------------------------------------------------------------------------------------
\subsection{Search versus certification: a structural decomposition}\label{subsec:search_cert_decomp}

The stopping-time viewpoint becomes especially transparent once we separate two qualitatively different contributions to $T$~\cite{Siegmund2013,Bang2025securePAC}.
\begin{itemize}
\item[1.] {\bf Search}: The time required for the feedback-driven updates to steer $\mathbf{p}^{(n)}$ into a regime where
$F(\mathbf{p})$ is sufficiently close to unity.
\item[2.] {\bf Certification}: The additional time required to \emph{observe} $M_H$ consecutive successes once such a
regime is reached.
\end{itemize}

To formalize this decomposition, fix $\varepsilon\in(0,1)$ and define the high-fidelity set
\begin{eqnarray}
\mathcal{G}_{\varepsilon} := \big\{\mathbf{p} \in \mathcal{P}: \epsilon(\mathbf{p})\le \varepsilon\big\}
%	&=& \big\{\mathbf{p}\in\mathcal{P}: F(\mathbf{p})\ge 1-\varepsilon\big\},
\label{eq:good_set}
\end{eqnarray}
together with the hitting time
\begin{eqnarray}
\tau_{\varepsilon}
:=\inf\big\{n\ge 1: \mathbf{p}^{(n)}\in\mathcal{G}_{\varepsilon}\big\}.
\label{eq:hitting_time_tau}
\end{eqnarray}
Once the process enters $\mathcal{G}_{\varepsilon}$ and stays there until halting, the remaining sample overhead is controlled by run statistics of (approximately) Bernoulli trials. This universal certification law will be developed later (in Sec.~\ref{sec:noiseless_bounds}) and will serve as the dimension-independent component of all sample-complexity bounds.

%---------------------------------------------------------------------------------------------------------------------------------------------------------------------------------
\subsection{Dimension dependence and the role of noise}\label{subsec:dimension_and_noise}

The decomposition above clarifies how the two key parameters enter the sample complexity.
\begin{itemize}
\item {\em Dependence on $M_H$.}---The halting threshold sets the certificate strength through $\epsilon_{\rm cert}(M_H,\delta)$ in Eq.~(\ref{eq:eps_cert}), but longer certificates necessarily require more samples to be observed.
\item {\em Dependence on $d$.}---The Hilbert-space dimension influences primarily the \emph{search} stage through the geometry and effective dimension of the control landscape for $\hat{U}(\mathbf{p})$. In contrast, once a uniform lower bound on $F$ is available, the certification cost is governed by universal run statistics and is largely independent of $d$.
\end{itemize}

Finally, the stopping-time viewpoint anticipates a distinctive robustness phenomenon. If the recorded label is corrupted by a noisy one-bit channel (classification noise), the observable success probability is capped away from $1$ even when perfect control is available. In that case, long success runs become exponentially rare and the halting time can blow up sharply; this mechanism will be quantified in Sec.~\ref{sec:noise_analysis}.

\begin{remark}[Scope of the stopping-time criteria]\label{rem:scope_perf}
The definitions of $T$, $P(N)$, and $T_{\delta}$ are purely operational and do not rely on a particular choice of update distribution $\mathbf{r}_{n}$ or a specific parametrization of $\hat{U}(\mathbf{p})$. They therefore provide a common language for comparing SSML variants and for isolating which aspects of performance are algorithmic (search dynamics) and which are statistical (certification).
\end{remark}

%-------------------------------------------------------------------------------------------------------------------------------------------------------------------------------------------------------------------------------------
\section{Noiseless analysis: search dynamics and certification cost}\label{sec:noiseless_bounds}
%-------------------------------------------------------------------------------------------------------------------------------------------------------------------------------------------------------------------------------------

In the noiseless setting, the only randomness in SSML arises from the intrinsic Born rule of the binary measurement and from the algorithmic randomness injected through the update directions $\mathbf{r}_{n}$. This section develops the technical core of our stopping-time theory in two steps. First, we quantify the certification cost induced by the run-length halting rule using universal run statistics. Second, we analyze the search dynamics that drive the success probability toward unity and combine the two ingredients into explicit sample-complexity bounds.

A recurring structural feature is that SSML is intrinsically ``two-speed'': the parameter changes only upon failures, while successes freeze the control and generate long dwell times precisely in high-fidelity regions ({\bf Remark~\ref{rem:ssml_freezing}}). This time-scale separation is what makes stopping-time analysis both operational and mathematically tractable.

%---------------------------------------------------------------------------------------------------------------------------------------------------------------------------------
\subsection{Universal certification law: runs of successes}\label{subsec:noiseless_certification}

We begin with the certification mechanism. Consider a fixed control $\mathbf{p}$ and repeated application of the binary test $\{\hat{M}_{f}, \hat{M}_{f^{\perp}}\}$. The resulting record is an i.i.d. Bernoulli process with success probability $p=F(\mathbf{p})$. Let $X_k(p)$ denote the waiting time until the first appearance of $k$ consecutive successes~\cite{Drekic2021}.

\begin{theorem}[Run waiting time: mean and tail]\label{thm:runs_mean_tail}
Let $\{B_n\}_{n\ge 1}$ be i.i.d.\ Bernoulli$(p)$ with $p\in(0,1)$ and define
\begin{eqnarray}
X_k(p) := \inf\big\{n\ge 1: B_{n-k+1} = \cdots = B_n = 1 \big\},
\label{eq:runs_waiting_time_def}
\end{eqnarray}
the waiting time for $k$ consecutive successes. Then,
\begin{eqnarray}
\mathbb{E}\!\big[X_k(p)\big] = \frac{1-p^{k}}{(1-p)p^{k}}.
\label{eq:main_Ek}
\end{eqnarray}
Moreover, for any $N\in\mathbb{N}$,
\begin{eqnarray}
\mathbb{P}\big(X_k(p) > N\big) \le \bigl(1-p^{k}\bigr)^{\left\lfloor N/k \right\rfloor}.
\label{eq:main_tail_bound}
\end{eqnarray}
\end{theorem}

\begin{proof}---The proof is two-fold.

\smallskip
\noindent {\bf Mean.} Let $E_j$ be the expected additional number of trials needed to observe $k$ consecutive successes given that the current run length is $j \in \{0,1,\dots,k\}$. Clearly $E_k=0$. For $0 \le j \le k-1$, one step is consumed and then the run length either increments to $j+1$ with probability $p$ or resets to $0$ with probability $1-p$, giving
\begin{eqnarray}
E_{j} = 1 + p E_{j+1} + (1-p) E_{0}.
\label{eq:cert_Ej_recursion}
\end{eqnarray}
Subtracting the $j=0$ equation from the general $j$ equation eliminates $E_0$ and yields a telescoping relation. Solving the resulting linear system gives
\begin{eqnarray}
E_0 = \frac{1-p^k}{(1-p)p^k},
\end{eqnarray}
which proves Eq.~(\ref{eq:main_Ek}).

\smallskip
{\bf Tail (block bound).} Partition the first $N$ trials into $B:=\lfloor N/k\rfloor$ disjoint blocks of length $k$. If a run of $k$ consecutive successes occurs within the first $N$ trials, then in particular at least one block contains $k$ successes. Hence,
\begin{eqnarray}
\mathbb{P}\big(X_{k}(p) > N\big) &\le& \mathbb{P}\Big(\text{no block is all-success}\Big) \nonumber \\
	&=& \bigl( 1-p^{k} \bigr)^{B},
\label{eq:cert_tail_block}
\end{eqnarray}
where independence holds because we used disjoint blocks. This proves Eq.~(\ref{eq:main_tail_bound}).
\end{proof}

{\bf Theorem~\ref{thm:runs_mean_tail}} is the mathematical backbone of the halting rule. Setting $k=M_H$ and $p=1-\varepsilon$ yields the canonical scaling
\begin{eqnarray}
\mathbb{E}\!\big[X_{M_H}(1-\varepsilon)\big] = \frac{1-(1-\varepsilon)^{M_H}}{\varepsilon(1-\varepsilon)^{M_H}} \approx \frac{e^{\varepsilon M_H}-1}{\varepsilon},
\label{eq:cert_scaling_recall_V}
\end{eqnarray}
together with a high-confidence implication from Eq.~(\ref{eq:main_tail_bound}).

\begin{remark}[Linear versus exponential certification]\label{rem:linear_vs_exp_cert}
The dimensionless product $\varepsilon M_H$ governs certification. If $\varepsilon M_H=O(1)$ then $(1-\varepsilon)^{M_H}=\Theta(1)$ and $\mathbb{E}[X_{M_H}(1-\varepsilon)]=\Theta(M_H)$. If $\varepsilon M_H\gg 1$ then $(1-\varepsilon)^{M_H}$ is exponentially small and the mean waiting time grows essentially as $\exp(\varepsilon M_H)$. Operationally, a strict run-length certificate is feasible only when the underlying success probability is within $O(1/M_H)$ of unity.
\end{remark}

%---------------------------------------------------------------------------------------------------------------------------------------------------------------------------------
\subsection{Universal stopping-time decomposition: search versus certification}\label{subsec:noiseless_decomposition}

We now connect the universal run statistics to SSML. Fix $\varepsilon\in(0,1)$ and recall the high-fidelity set $\mathcal{G}_{\varepsilon}$ and its hitting time $\tau_{\varepsilon}$ from Eq.~(\ref{eq:good_set}) and Eq.~(\ref{eq:hitting_time_tau}). Intuitively, $\tau_{\varepsilon}$ captures the search effort required to reach a region where the certificate is meaningful, while the additional overhead to halt is governed by $X_{M_H}(1-\varepsilon)$.

\begin{theorem}[Universal decomposition bound]\label{thm:decomposition_bound}
Fix $\varepsilon\in(0,1)$. Suppose that, on the event $\{\tau_{\varepsilon}<\infty\}$, the SSML trajectory satisfies $F(\mathbf{p}^{(n)})\ge 1-\varepsilon$ for all $n\ge \tau_{\varepsilon}$ until halting. Then, the stopping time $T$ satisfies the stochastic domination bound
\begin{eqnarray}
T \le \tau_{\varepsilon} + X_{M_H}(1-\varepsilon)
\label{eq:main_stoch_dom}
\end{eqnarray}
in stochastic dominance~\cite{Shaked2007}. In particular,
\begin{eqnarray}
\mathbb{E}[T] \le \mathbb{E}\!\big[\tau_{\varepsilon}\big] + \frac{1-(1-\varepsilon)^{M_H}}{\varepsilon(1-\varepsilon)^{M_H}}.
\label{eq:main_ET_bound}
\end{eqnarray}
Moreover, for any $\delta\in(0,1)$,
\begin{eqnarray}
T_{\delta} \le \tau_{\varepsilon, \delta/2} + \frac{M_H}{(1-\varepsilon)^{M_H}}\ln\frac{2}{\delta},
\label{eq:main_Tdelta_bound}
\end{eqnarray}
where $\tau_{\varepsilon,\delta}$ denotes the $(1-\delta)$-quantile of $\tau_{\varepsilon}$.
\end{theorem}

\begin{proof}---On the event $\{\tau_{\varepsilon} < \infty\}$, let us consider the post-hitting record $\{m_{\tau_{\varepsilon}+1},m_{\tau_{\varepsilon}+2},\dots\}$. By assumption, each trial has conditional success probability at least $1-\varepsilon$ given the past. Let $\{U_n\}_{n\ge 1}$ be i.i.d.\ uniform on $[0,1]$ and couple the SSML record after $\tau_{\varepsilon}$ by setting
\begin{eqnarray}
m_{\tau_{\varepsilon}+n}=s \ \Longleftrightarrow\ U_n \le F\big(\mathbf{p}^{(\tau_{\varepsilon}+n)}\big).
\label{eq:coupling_uniform}
\end{eqnarray}
Define an auxiliary i.i.d.\ Bernoulli$(1-\varepsilon)$ process $B_n:=\mathbb{1}_{\{U_n\le 1-\varepsilon\}}$. Since $F(\mathbf{p}^{(\tau_{\varepsilon}+n)})\ge 1-\varepsilon$, we have $\{B_n=1\} \subseteq \{m_{\tau_{\varepsilon}+n}=s\}$ for every $n$ under this coupling. Therefore, whenever the auxiliary process has a run of $M_H$ consecutive successes, so does the SSML record, and the SSML halting time after $\tau_{\varepsilon}$ is stochastically dominated by $X_{M_H}(1-\varepsilon)$. This proves Eq.~(\ref{eq:main_stoch_dom}).

Taking expectations yields Eq.~(\ref{eq:main_ET_bound}) by {\bf Theorem~\ref{thm:runs_mean_tail}}. For the quantile bound, write for any $N$,
\begin{eqnarray}
&& \mathbb{P}(T>N) \nonumber \\
&& \quad \le \mathbb{P}\big(\tau_{\varepsilon}>N_1\big) + \mathbb{P}\big(X_{M_H}(1-\varepsilon)>N-N_1\big),
\label{eq:quantile_union_bound}
\end{eqnarray}
and choose $N_1:=\tau_{\varepsilon,\delta/2}$ so that the first term is at most $\delta/2$. For the second term, apply Eq.~(\ref{eq:main_tail_bound}) with $k=M_H$ and solve for $N-N_1$ so that the right-hand side is at most $\delta/2$, yielding Eq.~(\ref{eq:main_Tdelta_bound}).
\end{proof}

{\bf Theorem~\ref{thm:decomposition_bound}} makes explicit what is universal and what is model dependent: the certification term is controlled solely by run statistics, while all dependence on the control landscape and on dimension $d$ is relegated to the search time $\tau_{\varepsilon}$.

%---------------------------------------------------------------------------------------------------------------------------------------------------------------------------------
\subsection{Run length as a monitored proxy for infidelity}\label{subsec:noiseless_MS_proxy}

The SSML literature emphasizes that the run-length counter $M_S$ can be used as a real-time indicator of accuracy. We record the precise probabilistic content because it will be used repeatedly in the search analysis.

Fix a parameter $\mathbf{p}$ and consider the number $L(\mathbf{p})$ of consecutive successes observed before the next failure when measuring with success probability $F(\mathbf{p})=1-\epsilon(\mathbf{p})$. Then,
\begin{eqnarray}
\mathbb{P}\big(L(\mathbf{p})=\ell \big| \mathbf{p}\big) = \bigl(1-\epsilon(\mathbf{p})\bigr)^{\ell} \epsilon(\mathbf{p}),
\label{eq:geom_run_dist}
\end{eqnarray}
for $\ell=0,1,2,\dots$, i.e., $L(\mathbf{p})$ is geometric with mean
\begin{eqnarray}
\mathbb{E}\big[L(\mathbf{p}) \big| \mathbf{p}\big] = \frac{1-\epsilon(\mathbf{p})}{\epsilon(\mathbf{p})} = \epsilon(\mathbf{p})^{-1}-1.
\label{eq:geom_mean}
\end{eqnarray}
Therefore, whenever the process is locally stationary so that $\mathbf{p}$ is effectively fixed over many shots, one obtains the monitored proxy
\begin{eqnarray}
\epsilon(\mathbf{p}) \approx \frac{1}{1+M_S}.
\label{eq:MS_proxy}
\end{eqnarray}
This heuristic relation is the basis for the empirically successful choice $\beta \simeq 1/2$ in SSML implementations~\cite{Lee2018}.

%---------------------------------------------------------------------------------------------------------------------------------------------------------------------------------
\subsection{Search dynamics: a failure-driven skeleton chain}\label{subsec:noiseless_search_skeleton}

The search component is model dependent because it depends on how $\hat{U}(\mathbf{p})$ is parametrized and how the random directions $\mathbf{r}_{n}$ are chosen. Nevertheless, SSML admits a useful universal reformulation in terms of a failure-driven skeleton.

Define the sequence of failure times
\begin{eqnarray}
\tau_{0}:=0, \quad \tau_{k+1}:=\inf\big\{n>\tau_{k}: m_{n}=u\big\},
\label{eq:failure_times}
\end{eqnarray}
for $k \ge 0$, and define the skeleton parameters right after each update,
\begin{eqnarray}
\mathbf{P}_{k} := \mathbf{p}^{(\tau_{k}+1)}.
\label{eq:skeleton_param_def}
\end{eqnarray}
Between $\tau_{k}$ and $\tau_{k+1}$ the parameter is frozen and the process simply produces a run of successes of random length $L_{k}:=\tau_{k+1}-\tau_{k}-1$. Conditioned on $\mathbf{P}_{k}$, the run length $L_{k}$ is geometric with parameter $\epsilon(\mathbf{P}_{k})$ as in Eq.~(\ref{eq:geom_run_dist}).

This representation makes two key facts explicit.
\begin{itemize}
\item[(i)] {\em Time is reweighted by fidelity.}---A parameter with smaller $\epsilon(\mathbf{P}_{k})$ generates longer success runs and thus occupies a larger fraction of physical time (number of consumed copies). This is the mechanism by which SSML ``prefers'' high-fidelity parameters without explicitly computing gradients.
\item[(ii)]{\em Updates are scale-adaptive.}---The step size at the $k$th update is $\omega_{k}=\alpha(L_{k}+1)^{-\beta}$, so in the high-fidelity regime Eq.~(\ref{eq:geom_mean}) suggests $\omega_{k}\sim \alpha\,\epsilon(\mathbf{P}_{k})^{\beta}$. In particular, for $\beta=1/2$,
\begin{eqnarray}
\omega_{k} \sim \alpha \sqrt{\epsilon(\mathbf{P}_{k})}.
\label{eq:omega_sqrt_eps}
\end{eqnarray}
This scale matching is at the heart of the near-linear sample complexity in the noiseless regime.
\end{itemize}

%---------------------------------------------------------------------------------------------------------------------------------------------------------------------------------
\subsection{Local refinement in the quadratic regime}\label{subsec:noiseless_local_refinement}

To translate the skeleton representation into quantitative bounds, we focus on the refinement stage that dominates for large $M_H$: the regime where $\epsilon(\mathbf{p}) \ll 1$ and the infidelity landscape is approximately quadratic.

\begin{assumption}[Random perturbation model]\label{assump:random_perturb}
The update directions $\{\mathbf{r}_{n}\}$ are i.i.d. random vectors in an effective parameter dimension $m$. They are isotropic and satisfy $\norm{\mathbf{r}_{n}}=1$ almost surely (or are uniformly bounded with an isotropic angular distribution). Moreover, conditioned on the current control $\mathbf{p}^{(n)}$, the draw of $\mathbf{r}_{n}$ is independent of the Born-rule randomness in the subsequent measurement outcome.
\end{assumption}

\begin{assumption}[Locally quadratic infidelity and bounded refinement]\label{assump:local_quadratic}
There exist constants $c_{-}, c_{+} > 0$ and a neighborhood $\mathcal{U}$ of an optimal parameter $\mathbf{p}_{\star}$ such that, for all $\mathbf{p}\in\mathcal{U}$,
\begin{eqnarray}
c_{-} \norm{\mathbf{p}-\mathbf{p}_{\star}}^{2} \le \epsilon(\mathbf{p}) \le c_{+} \norm{\mathbf{p}-\mathbf{p}_{\star}}^{2}.
\label{eq:main_local_quad}
\end{eqnarray}
We assume that during the refinement stage the SSML trajectory remains inside $\mathcal{U}$.
\end{assumption}

{\bf Assumption~\ref{assump:local_quadratic}} formalizes the standard intuition that near an optimum the infidelity is quadratic in the control error, and it is consistent with the heuristic choice $\beta \simeq 1/2$. Under these assumptions, one can prove a scale-invariant one-step improvement event.

\begin{lemma}[Scale-invariant local improvement]\label{lem:local_improvement_event}
Assume $\beta=1/2$ and {\bf Assumptions~\ref{assump:random_perturb}} and {\bf \ref{assump:local_quadratic}}. There exist constants $\kappa \in (0,1)$ and $\rho \in (0,1)$, depending only on
$(\alpha,m,c_{-},c_{+})$ and on the distribution of $\mathbf{r}_{n}$, such that for every $\mathbf{p}\in\mathcal{U}$ with $\epsilon(\mathbf{p})$ sufficiently small, the next failure-triggered update satisfies
\begin{eqnarray}
\mathbb{P}\left( \epsilon(\mathbf{p}^{+}) \le (1-\kappa)\epsilon(\mathbf{p}) \mid \mathbf{p} \right) \ge \rho,
\label{eq:local_improve_prob}
\end{eqnarray}
where $\mathbf{p}^{+}$ denotes the post-update parameter.
\end{lemma}

\begin{proof}---Fix $\mathbf{p}\in\mathcal{U}$ and write $\mathbf{x}:=\mathbf{p}-\mathbf{p}_{\star}$. Let $\varepsilon:=\epsilon(\mathbf{p})$. Conditioned on $\mathbf{p}$, the run length $L$ before the next failure is geometric with parameter $\varepsilon$ (Eq.~(\ref{eq:geom_run_dist})), and the subsequent update has the form $\mathbf{p}^{+}=\mathbf{p}+\omega\,\mathbf{r}$ with $\omega=\alpha(L+1)^{-1/2}$ and $\mathbf{r}$ isotropic ({\bf Assumption~\ref{assump:random_perturb}}).

\smallskip
\noindent {\bf Step-size window.} Fix constants $0 < a < b$ and define the event
\begin{eqnarray}
\mathcal{E}_{L} := \Big\{ \Big\lceil \frac{a}{\varepsilon} \Big\rceil \le L \le \Big\lfloor \frac{b}{\varepsilon} \Big\rfloor \Big\}.
\label{eq:event_EL_def}
\end{eqnarray}
Using $\mathbb{P}(L \ge \ell \mid \mathbf{p}) = (1-\varepsilon)^{\ell}$, we have
\begin{eqnarray}
\mathbb{P}(\mathcal{E}_{L}\mid \mathbf{p}) = (1-\varepsilon)^{\lceil a/\varepsilon\rceil} - (1-\varepsilon)^{\lfloor b/\varepsilon\rfloor+1}.
\label{eq:EL_prob_exact}
\end{eqnarray}
For $\varepsilon$ sufficiently small, this probability is bounded below by a positive constant depending only on $(a,b)$ (indeed it converges to $e^{-a}-e^{-b}$ as $\varepsilon\to 0$). On $\mathcal{E}_L$, the step size satisfies
\begin{eqnarray}
\alpha (b/\varepsilon+1)^{-1/2} \le \omega \le \alpha(a/\varepsilon)^{-1/2},
\label{eq:omega_window}
\end{eqnarray}
hence, $\omega=\Theta(\sqrt{\varepsilon})$.

\smallskip
\noindent {\bf Directional improvement.} Let $\hat{\mathbf{x}}:=\mathbf{x}/\norm{\mathbf{x}}$ and define $U:=\hat{\mathbf{x}} \cdot \mathbf{r}$. By isotropy, $U$ is a symmetric random variable whose distribution depends only on $m$. Fix $\gamma \in (0,1)$ and let $\mathcal{E}_{r} := \{U \le -\gamma\}$. Then, $\mathbb{P}(\mathcal{E}_{r}) = \mu_m(\gamma) > 0$, the measure of a spherical cap.

On $\mathcal{E}_L \cap \mathcal{E}_r$, we can bound the post-update distance:
\begin{eqnarray}
\norm{\mathbf{x}^{+}}^{2} &=& \norm{\mathbf{x}+\omega\mathbf{r}}^{2} \nonumber \\
	&=& \norm{\mathbf{x}}^{2}+\omega^{2}+2\omega\norm{\mathbf{x}} U \nonumber \\
	&\le& \norm{\mathbf{x}}^{2}\Big(1+\lambda^{2}-2\lambda\gamma\Big),
\label{eq:xplus_bound}
\end{eqnarray}
where $\lambda:=\omega/\norm{\mathbf{x}}$. By {\bf Assumption~\ref{assump:local_quadratic}}, $\norm{\mathbf{x}}^{2}\le \varepsilon/c_{-}$ and $\norm{\mathbf{x}}^{2}\ge \varepsilon/c_{+}$. By combining with Eq.~(\ref{eq:omega_window}), we can yield constants $\lambda_{\min},\lambda_{\max} > 0$ (depending only on $\alpha, a, b, c_{\pm}$) such that $\lambda \in [\lambda_{\min}, \lambda_{\max}]$ on $\mathcal{E}_L$. Choose $(a,b)$ and $\gamma$ so that
\begin{eqnarray}
\frac{c_{+}}{c_{-}}\Big(1+\lambda_{\max}^{2}-2\lambda_{\min}\gamma\Big) \le 1-\kappa
\label{eq:choose_kappa_condition}
\end{eqnarray}
for some $\kappa\in(0,1)$. Then, on $\mathcal{E}_L \cap \mathcal{E}_r$, we have
\begin{eqnarray}
\epsilon(\mathbf{p}^{+}) \le c_{+}\norm{\mathbf{x}^{+}}^{2} &\le& (1-\kappa)c_{-}\norm{\mathbf{x}}^{2} \nonumber \\
	&\le& (1-\kappa)\epsilon(\mathbf{p}),
\end{eqnarray}
which proves the contraction event.

Finally, by construction and conditional independence,
\begin{eqnarray}
\mathbb{P}(\mathcal{E}_L \cap \mathcal{E}_r \mid \mathbf{p}) = \mathbb{P}(\mathcal{E}_L \mid \mathbf{p}) \mathbb{P}(\mathcal{E}_r) \ge \rho,
\end{eqnarray}
for some $\rho > 0$ uniform over all sufficiently small $\varepsilon$ in $\mathcal{U}$, where $\rho$ depends only on $(\alpha, m, c_{\pm})$ and the choice of $(a, b, \gamma)$. This proves Eq.~(\ref{eq:local_improve_prob}).
\end{proof}

{\bf Lemma~\ref{lem:local_improvement_event}} formalizes an important conceptual point: for $\beta=1/2$, the update magnitude matches the local error scale, so the chance of making a constant-factor
improvement does not vanish as $\epsilon\to 0$.

We now convert one-step improvements into a refinement bound.

\begin{proposition}[Local refinement cost]\label{prop:local_refinement}
Assume the setting of {\bf Lemma~\ref{lem:local_improvement_event}}. Fix $\epsilon_{0} \in (0,1)$, such that $\{\epsilon(\mathbf{p}) \le \epsilon_{0}\} \subset \mathcal{U}$. Starting from any $\mathbf{p} \in \mathcal{U}$ with $\epsilon(\mathbf{p}) \le \epsilon_{0}$, let $K$ be the number of failure-triggered updates needed to reach $\epsilon \le \epsilon_{\rm targ}$ for a target
$\epsilon_{\rm targ}\in(0,\epsilon_{0})$.
Then,
\begin{eqnarray}
\mathbb{E}[K] \le \frac{1}{\rho} \Bigg\lceil \frac{\ln(\epsilon_{0}/\epsilon_{\rm targ})}{\abs{\ln(1-\kappa)}} \Bigg\rceil.
\label{eq:EK_bound}
\end{eqnarray}
Moreover, the expected number of consumed copies during this refinement stage satisfies
\begin{eqnarray}
\mathbb{E}[N_{\rm ref}] \le \frac{1}{\epsilon_{\rm targ}}\,\mathbb{E}[K] \le \frac{1}{\rho\epsilon_{\rm targ}}\Bigg\lceil \frac{\ln(\epsilon_{0}/\epsilon_{\rm targ})}{\abs{\ln(1-\kappa)}} \Bigg\rceil.
\label{eq:ENref_bound}
\end{eqnarray}
\end{proposition}

\begin{proof}---Define an update as ``successful'' if it produces the contraction event $\epsilon(\mathbf{p}^{+}) \le (1-\kappa)\epsilon(\mathbf{p})$. By {\bf Lemma~\ref{lem:local_improvement_event}}, each update is successful with conditional probability at least $\rho$. Using the uniform coupling argument from the proof of {\bf Theorem~\ref{thm:decomposition_bound}}, the number of successful contractions after $k$ updates stochastically dominates a Binomial$(k,\rho)$ random variable, and hence the expected number of updates required to obtain $s$ successful contractions is at most $s/\rho$.

Each successful contraction reduces $\epsilon$ by at least a factor $(1-\kappa)$, so it suffices to obtain
\begin{eqnarray}
s \ge \frac{\ln(\epsilon_{0}/\epsilon_{\rm targ})}{\abs{\ln(1-\kappa)}}
\label{eq:needed_contractions}
\end{eqnarray}
successful contractions. This proves Eq.~(\ref{eq:EK_bound}).

For the shot cost, note that before reaching $\epsilon_{\rm targ}$ we always have $\epsilon(\mathbf{P}_{k}) \ge \epsilon_{\rm targ}$ at every skeleton parameter in the refinement stage. Conditioned on $\mathbf{P}_{k}$, the number of consumed copies between two successive updates equals $L_{k}+1$ and has mean $\mathbb{E}[L_{k}+1 \mid \mathbf{P}_{k}]=1/\epsilon(\mathbf{P}_{k}) \le 1/\epsilon_{\rm targ}$ (cf. Eq.~(\ref{eq:geom_mean})).
Therefore, by linearity of expectation,
\begin{eqnarray}
\mathbb{E}[N_{\rm ref}] = \mathbb{E}\Big[\sum_{k=0}^{K-1}(L_k+1)\Big] \le \frac{1}{\epsilon_{\rm targ}}\,\mathbb{E}[K],
\end{eqnarray}
which gives Eq.~(\ref{eq:ENref_bound}).
\end{proof}

\begin{remark}[Refinement cost and the $1/\epsilon$ law]\label{rem:ref_cost_intuition}
The scaling $\mathbb{E}[L+1\mid \epsilon]=1/\epsilon$ is a direct consequence of freezing on success: probing a high-fidelity control consumes more copies precisely because long success runs occur. Thus, any sequential scheme that uses run-length evidence is naturally expected to exhibit a $1/\epsilon$ resource law.
\end{remark}

%---------------------------------------------------------------------------------------------------------------------------------------------------------------------------------
\subsection{Noiseless stopping-time sample complexity}\label{subsec:noiseless_combine}

We now connect the refinement analysis to the stopping time $T$. A principled target for the search stage is the certificate scale. For a desired significance $\delta$, Eq.~(\ref{eq:eps_cert}) suggests choosing
\begin{eqnarray}
\epsilon_{\rm targ} \asymp \epsilon_{\rm cert}(M_H,\delta) \approx \frac{\ln(1/\delta)}{M_H},
\label{eq:eps_targ_choice}
\end{eqnarray}
because this keeps the certification cost in the linear regime of {\bf Remark~\ref{rem:linear_vs_exp_cert}}.

By combining {\bf Proposition~\ref{prop:local_refinement}} with the decomposition bound of {\bf Theorem~\ref{thm:decomposition_bound}}, we can consider an explicit noiseless stopping-time upper bound as below.
\begin{theorem}[Noiseless stopping-time upper bound: scaling law]\label{thm:noiseless_main}
Assume the reachability condition {\rm ({\bf Assumption~\ref{assump:ssml_reachability}})} and the local refinement assumptions {\rm ({\bf Assumptions~\ref{assump:random_perturb}} and {\bf \ref{assump:local_quadratic}})} with $\beta=1/2$. Fix $\varepsilon\in(0,\epsilon_{0})$ and suppose that once the SSML trajectory enters $\mathcal{G}_{\varepsilon}$ it remains in $\mathcal{U}$ and satisfies $F(\mathbf{p}^{(n)})\ge 1-\varepsilon$ until halting. Then,
\begin{eqnarray}
\mathbb{E}[T] \le \mathbb{E}\big[\tau_{\varepsilon}\big] + \frac{1-(1-\varepsilon)^{M_H}}{\varepsilon(1-\varepsilon)^{M_H}},
\label{eq:noiseless_ET_general}
\end{eqnarray}
and, in the refinement-dominated regime where $\varepsilon$ is reached through the local dynamics starting from $\epsilon(\mathbf{p})\le \epsilon_{0}$,
\begin{eqnarray}
\mathbb{E}\big[\tau_{\varepsilon}\big] \le \mathbb{E}\big[\tau_{\epsilon_{0}}\big] + \frac{1}{\rho \varepsilon} \Bigg\lceil \frac{\ln(\epsilon_{0}/\varepsilon)}{\abs{\ln(1-\kappa)}} \Bigg\rceil,
\label{eq:noiseless_search_bound}
\end{eqnarray}
where $(\rho,\kappa)$ are the constants from {\bf Lemma~\ref{lem:local_improvement_event}}. In particular, choosing $\varepsilon=\Theta(1/M_H)$ places certification in the linear regime and yields the scaling
\begin{eqnarray}
\mathbb{E}[T] = O\Big( \mathbb{E}[\tau_{\epsilon_{0}}] + \mathsf{K}(d) M_H\ln{M_H} + M_H \Big),
\label{eq:noiseless_scaling_summary}
\end{eqnarray}
where $\mathsf{K}(d)$ denotes the dimension-dependent inverse improvement probability $\mathsf{K}(d)\asymp 1/\rho$.
\end{theorem}

\begin{proof}---Eq.~(\ref{eq:noiseless_ET_general}) is Eq.~(\ref{eq:main_ET_bound}) in {\bf Theorem~\ref{thm:decomposition_bound}}. For Eq.~(\ref{eq:noiseless_search_bound}), consider the refinement stage starting from the first time the trajectory reaches $\epsilon \le \epsilon_{0}$. By {\bf Proposition~\ref{prop:local_refinement}} with $\epsilon_{\rm targ}=\varepsilon$, the expected number of samples
required to reach $\epsilon \le \varepsilon$ from $\epsilon \le \epsilon_{0}$ is bounded by the second term on the right-hand side of Eq.~(\ref{eq:noiseless_search_bound}). Adding the pre-refinement cost $\mathbb{E}[\tau_{\epsilon_{0}}]$ yields Eq.~(\ref{eq:noiseless_search_bound}). The scaling statement Eq.~(\ref{eq:noiseless_scaling_summary}) follows by choosing $\varepsilon=\Theta(1/M_H)$ and using $\ln(\epsilon_{0}/\varepsilon)=\Theta(\ln M_H)$.
\end{proof}

\begin{remark}[Dimension dependence]\label{rem:dimension_dependence_V}
In the bounds above, all dependence on the Hilbert-space dimension enters through the \emph{search} stage, encoded in the improvement probability $\rho$, hence $\mathsf{K}(d)\asymp 1/\rho$, and the pre-refinement time $\mathbb{E}[\tau_{\epsilon_{0}}]$. Once a uniform lower bound $F \ge 1-\varepsilon$ is available, the certification term is governed by universal run statistics {\rm ({\bf Theorem~\ref{thm:runs_mean_tail}})} and is essentially independent of $d$. This separation is the stopping-time manifestation of the ``search versus certificate'' architecture of SSML.
\end{remark}

%-------------------------------------------------------------------------------------------------------------------------------------------------------------------------------------------------------------------------------------
\section{Classification-noise analysis: feasibility and blow-up}\label{sec:noise_analysis}
%-------------------------------------------------------------------------------------------------------------------------------------------------------------------------------------------------------------------------------------

The halting rule in SSML is intrinsically a run-length certificate: the algorithm terminates only after observing $M_H$ consecutive ``success'' outcomes. This is precisely the feature that makes SSML operationally attractive in high-precision regimes, because the certificate is produced online from a single-bit-per-copy record. However, the same feature makes SSML particularly sensitive to classification noise---imperfections that flip the recorded success/failure labels~\cite{Angluin1988,Cover1999}. In this section, we quantify this sensitivity and establish a sharp feasibility picture: once the noise probability $q$ is such that $qM_H$ is not small, the stopping time becomes exponentially large and ``learning completion'' is effectively unattainable within any realistic shot budget.

A key point that we emphasize throughout is that the classification noise impacts SSML in two logically distinct ways: (i) search degradation through incorrect feedback updates, and (ii) certification infeasibility through the rarity of long success runs.

%Our main results in this section concern (ii): they are \emph{universal} in the sense that they hold regardless of the details of the search dynamics. In particular, even an idealized ``oracle'' that always chooses the best control cannot evade the exponential certification overhead induced by noisy labels.

%---------------------------------------------------------------------------------------------------------------------------------------------------------------------------------
\subsection{Noise model: a binary channel on the measurement record}\label{subsec:noise_model}

We formalize the classification noise as a corruption of the classical label produced by the measurement. Let $y_n \in \{s,u\}$ denote the true measurement outcome that would be obtained from the Born rule given the current control $\mathbf{p}^{(n)}$, and let $m_n \in \{s,u\}$ denote the recorded label that the feedback rule actually uses in Eq.~(\ref{eq:ssml_memory_update}) and Eq.~(\ref{eq:ssml_param_update}). We consider the two standard models~\cite{Cover1999}.

\medskip
\noindent [{\bf N.1}] {\em Binary-symmetric classification noise (BSC)}.
The recorded label is flipped with probability $q\in[0,1/2)$:
\begin{eqnarray}
\mathbb{P}(m_n\neq y_n)=q,\quad \mathbb{P}(m_n=y_n)=1-q,
\label{eq:BSC_def}
\end{eqnarray}
independently across steps and independently of all other randomness conditioned on the current true label. In this case, the effective observed success probability becomes
\begin{eqnarray}
\widetilde{F}(\mathbf{p}) &:=& \mathbb{P}(m_n=s\mid \mathbf{p}^{(n)}=\mathbf{p}) \nonumber \\
	&=& (1-q)F(\mathbf{p})+q\big(1-F(\mathbf{p})\big) \nonumber \\
	&=& q+(1-2q)F(\mathbf{p}).
\label{eq:effective_F_BSC}
\end{eqnarray}

\medskip
\noindent [{\bf N.2}] {\em False-negative noise (FN)}. A true success is recorded as a failure with probability $q \in [0,1)$, while failures are recorded faithfully:
\begin{eqnarray}
\mathbb{P}(m_n=u\mid y_n=s) &=& q, \nonumber \\
\mathbb{P}(m_n=s\mid y_n=s) &=& 1-q, \nonumber \\
\mathbb{P}(m_n=u\mid y_n=u) &=& 1.
\label{eq:FN_def}
\end{eqnarray}
Then,
\begin{eqnarray}
\widetilde{F}(\mathbf{p}) = (1-q) F(\mathbf{p}).
\label{eq:effective_F_FN}
\end{eqnarray}

Both models share the same structural implication that is most important for halting-time analysis: the observable success probability is uniformly bounded away from $1$ even when perfect control is possible. Indeed, since $F(\mathbf{p}) \le 1$,
\begin{eqnarray}
\widetilde{F}(\mathbf{p}) \le \widetilde{F}_{\max} := \sup_{\mathbf{p} \in \mathcal{P}}\widetilde{F}(\mathbf{p}) = 1-q
\label{eq:Fmax_noise}
\end{eqnarray}
for both [{\bf N.1}] and [{\bf N.2}]. This cap is the origin of the operational threshold $qM_H \gtrsim 1$.

\begin{remark}[Noise enters through the one-bit feedback channel]\label{rem:noise_channel}
The transformation $F(\mathbf{p}) \mapsto \widetilde{F}(\mathbf{p})$ is an explicit representation of the effective one-bit feedback channel from the quantum system to the controller. In the BSC model, the channel capacity vanishes at $q=1/2$~\cite{Cover1999}, where $\widetilde{F}(\mathbf{p})\equiv 1/2$ and the record becomes statistically independent of the state and the control. This is an identifiability threshold for learning. In contrast, our main focus below is an operational threshold that occurs already for $q \ll 1/2$ due to the run-length nature of certification.
\end{remark}

%---------------------------------------------------------------------------------------------------------------------------------------------------------------------------------
\subsection{Unavoidable certification overhead: a universal lower bound}\label{subsec:noise_unavoidable}

We now quantify the impact of label noise on the halting time. Recall that SSML halts when the recorded run length reaches $M_H$, i.e., at the stopping time $T = \inf\{ n \ge 1 \ : \ M_S^{(n)}=M_H \}$ with $M_S^{(n)}$ updated from $\{m_n\}$. The following statement is the formal expression of an intuitive but crucial fact: even if the search stage is perfect, noisy labels alone enforce an exponential certification overhead.

\begin{theorem}[Noise-induced blow-up]\label{thm:noise_blowup}
Assume the halting rule is unchanged and the recorded label is corrupted by either model [{\bf N.1}] or [{\bf N.2}] with parameter $q > 0$. Let $X_{M_H}(p)$ be the waiting time for $M_H$ consecutive successes in an i.i.d. Bernoulli$(p)$ process. Then the SSML halting time satisfies the universal lower bound
\begin{eqnarray}
\mathbb{E}[T] \ge \mathbb{E}\big[X_{M_H}(1-q)\big] = \frac{1-(1-q)^{M_H}}{q(1-q)^{M_H}}.
\label{eq:ET_noise_lower_proof}
\end{eqnarray}
In particular, for $qM_H \gg 1$,
\begin{eqnarray}
\mathbb{E}[T] \gtrsim \frac{1}{q}\,e^{qM_H}.
\label{eq:ET_noise_exp_proof}
\end{eqnarray}
\end{theorem}

\begin{proof}---Consider an idealized ``best-case'' completion scenario in which the controller is always at a parameter $\mathbf{p}$, satisfying $F(\mathbf{p})=1$ ({\bf Assumption~\ref{assump:ssml_reachability}}). Under either noise model [{\bf N.1}] or [{\bf N.2}], the recorded outcome sequence is then an i.i.d. Bernoulli process with success probability $1-q$, because the only remaining source of failure labels is label corruption. In this best-case scenario, the SSML halting time is exactly $X_{M_H}(1-q)$.

For any actual SSML trajectory (and, more generally, for any learning strategy driven by the same noisy record), the recorded success probability at each step is at most $\widetilde{F}_{\max}=1-q$ by Eq.~(\ref{eq:Fmax_noise}). Since the waiting time to obtain a run of $M_H$ consecutive successes is monotone decreasing in the per-trial success probability, the best-case i.i.d. process with success probability $1-q$ minimizes the expected halting time among all processes with success probability bounded by $1-q$. Therefore, $\mathbb{E}[T]$ cannot be smaller than $\mathbb{E}[X_{M_H}(1-q)]$, which equals Eq.~(\ref{eq:ET_noise_lower_proof}) by {\bf Theorem~\ref{thm:runs_mean_tail}}. The asymptotic form Eq.~(\ref{eq:ET_noise_exp_proof}) follows from $(1-q)^{M_H}\approx e^{-qM_H}$ for $q \ll 1$.
\end{proof}

\begin{remark}[Tightness]\label{rem:tightness_noise}
The lower bound in Eq.~(\ref{eq:ET_noise_lower_proof}) is tight in the sense that it is achieved (up to the irreducible noise) once the search dynamics reaches a regime where $F(\mathbf{p}) \approx 1$ and remains there. Thus, classification noise creates a certification bottleneck that persists even for an idealized, perfectly controllable learner.
\end{remark}

%---------------------------------------------------------------------------------------------------------------------------------------------------------------------------------
\subsection{Feasibility regimes and an operational threshold}\label{subsec:noise_feasibility}

{\bf Theorem~\ref{thm:noise_blowup}} shows that the dimensionless product $q M_H$ governs the fundamental overhead. To make this explicit, let us define
\begin{eqnarray}
p_{\max}:=1-q.
\label{eq:pmax_def}
\end{eqnarray}
Then,
\begin{eqnarray}
\mathbb{E}\big[X_{M_H}(p_{\max})\big] = \frac{1-p_{\max}^{M_H}}{(1-p_{\max})p_{\max}^{M_H}} \approx \frac{e^{qM_H}-1}{q}.
%	&=& \frac{1-(1-q)^{M_H}}{q(1-q)^{M_H}} \nonumber \\
\label{eq:EX_noise_scaling}
\end{eqnarray}

\smallskip
\noindent {\bf Regime I (feasible, linear).} If $qM_H \ll 1$, then $(1-q)^{M_H} \approx 1-qM_H$ and
\begin{eqnarray}
\mathbb{E}\big[X_{M_H}(1-q)\big] \approx M_H.
\label{eq:noise_linear_regime}
\end{eqnarray}
In this regime, the label noise does not qualitatively change the scaling of the halting time.

\smallskip
\noindent {\bf Regime II (infeasible, exponential).} If $qM_H \gg 1$, then $(1-q)^{M_H} \approx e^{-qM_H}$ is exponentially small and
\begin{eqnarray}
\mathbb{E}\!\big[X_{M_H}(1-q)\big] \approx \frac{1}{q} e^{qM_H}.
\label{eq:noise_exp_regime}
\end{eqnarray}
In this regime, increasing $M_H$ strengthens the certificate only at the price of an exponential increase in the required number of copies.

This motivates the following operational interpretation:

\begin{definition}[Operational threshold]\label{def:operational_threshold}
We say that SSML operates in a noise-feasible regime if $qM_H \lesssim 1$, and in a noise-infeasible regime if $qM_H \gtrsim 1$, where $\lesssim$ and $\gtrsim$ hide modest numerical constants (e.g., of order unity) that depend on the acceptable overhead factor in the expected or high-confidence halting time.
\end{definition}

The threshold in {\bf Definition~\ref{def:operational_threshold}} is ``sharp'' in the sense that the functional dependence is exponential in $qM_H$ rather than polynomial.

\begin{remark}[High-confidence feasibility]\label{rem:high_conf_feasibility}
A similar statement holds at the level of tails. Using the block bound in {\bf Theorem~\ref{thm:runs_mean_tail}} for $X_{M_H}(p_{\max})$,
\begin{eqnarray}
\mathbb{P}\bigl(X_{M_H}(p_{\max})>N \bigr) \le \bigl( 1-p_{\max}^{M_H} \bigr)^{\lfloor N/M_H \rfloor},
\label{eq:noise_tail_block}
\end{eqnarray}
one finds that achieving $\mathbb{P}(X_{M_H}(p_{\max})\le N) \ge 1-\delta$ requires, up to constants,
\begin{eqnarray}
N \gtrsim \frac{M_H}{p_{\max}^{M_H}}\ln\!\frac{1}{\delta} \approx M_H e^{qM_H} \ln{\frac{1}{\delta}}.
\label{eq:noise_quantile_scaling}
\end{eqnarray}
Thus the exponential dependence on $qM_H$ persists for high-confidence halting, not only for the mean.
\end{remark}

%---------------------------------------------------------------------------------------------------------------------------------------------------------------------------------
\subsection{Noisy certification semantics and a certificate ceiling}\label{subsec:noise_certificate_semantics}

In the noiseless case, the halting event is a direct certificate on $F(\mathbf{p}_{\rm est})$ because $\mathbb{P}(\text{$M_H$ successes}\mid \mathbf{p})=F(\mathbf{p})^{M_H}$ ({\bf Lemma~\ref{lem:certificate_strength}}). Under label noise, the same logic applies with the effective success probability $\widetilde{F}(\mathbf{p})$.

\begin{lemma}[Certificate strength under BSC noise]\label{lem:noise_certificate}
Assume the BSC model {\rm [{\bf N.1}]} with flip probability $q \in [0,1/2)$. For a fixed control $\mathbf{p}$, the probability of observing $M_H$ recorded successes consecutively is
\begin{widetext}
\begin{eqnarray}
\mathbb{P}\big(\text{$M_H$ recorded successes}\ \big|\ \mathbf{p}\big) = \widetilde{F}(\mathbf{p})^{M_H} = \Big(q+(1-2q)F(\mathbf{p})\Big)^{M_H} = \Big(1-q-(1-2q)\epsilon(\mathbf{p})\Big)^{M_H}.
\label{eq:noise_cert_prob}
\end{eqnarray}
Consequently, for any $\epsilon_{0}\in(0,1)$,
\begin{eqnarray}
\epsilon(\mathbf{p})\ge \epsilon_{0} \ \Longrightarrow \ \mathbb{P}\big(\text{$M_H$ recorded successes}\ \big|\ \mathbf{p}\big) \le \bigl(1-q-(1-2q)\epsilon_{0}\bigr)^{M_H}.
\label{eq:noise_cert_tail}
\end{eqnarray}
\end{widetext}
\end{lemma}

{\bf Lemma~\ref{lem:noise_certificate}} shows that the halting event remains a sequential hypothesis test, but its strength is fundamentally weakened by $q$. In particular, if one aims to certify $\epsilon(\mathbf{p}_{\rm est})<\epsilon_{0}$ at significance $\delta$, a sufficient condition is
\begin{eqnarray}
\bigl( 1-q-(1-2q)\epsilon_{0} \bigr)^{M_H} \le \delta.
\label{eq:noise_cert_condition}
\end{eqnarray}
Solving Eq.~(\ref{eq:noise_cert_condition}) for $\epsilon_{0}$ yields an explicit noise-aware certificate scale
\begin{eqnarray}
\epsilon_{0} \ge \epsilon^{(q)}_{\rm cert}(M_H,\delta) := \frac{1-q-\delta^{1/M_H}}{1-2q},
\label{eq:noise_eps_cert_exact}
\end{eqnarray}
with the small-$(q,\delta)$ approximation
\begin{eqnarray}
\epsilon^{(q)}_{\rm cert}(M_H,\delta) \approx \frac{1}{1-2q}\Big(\frac{\ln(1/\delta)}{M_H}-q\Big).
\label{eq:noise_eps_cert_approx}
\end{eqnarray}

Here, we note that Eq.~(\ref{eq:noise_eps_cert_approx}) highlights a central operational implication: A run-length certificate cannot meaningfully certify infidelity far below the label-noise level. Indeed, once $q \gtrsim \ln(1/\delta)/M_H$, the right-hand side of Eq.~(\ref{eq:noise_eps_cert_approx}) becomes non-positive, indicating that the halting event no longer provides a nontrivial guarantee beyond what is already limited by label noise (unless one increases $M_H$ further, which triggers the exponential blow-up in Sec.~\ref{subsec:noise_feasibility}).

\begin{remark}[Certificate ceiling and the $1/q$ run-length limit]\label{rem:certificate_ceiling}
The cap $\widetilde{F}_{\max}=1-q$ implies that runs longer than $\Theta(1/q)$ are exponentially suppressed. This provides a principled explanation of the empirically observed fact that the consecutive-success counter $M_S$ cannot reliably exceed the signal-to-noise ratio scale $\mathrm{SNR} \sim 1/q$ in single-shot implementations: the certificate itself becomes noise-limited before the underlying state fidelity does.
\end{remark}

%---------------------------------------------------------------------------------------------------------------------------------------------------------------------------------
\subsection{Interaction with search dynamics (qualitative remarks)}\label{subsec:noise_search_effect}

The results above are certification-limited and therefore universal. In practice, classification noise also perturbs the \emph{search} stage because SSML updates are triggered by the recorded label:
\begin{itemize}
\item[(i)] False negatives induce spurious resets $M_S\leftarrow 0$ and unnecessary parameter updates, causing the search trajectory to diffuse even near an optimum.
\item[(ii)] False positives (in the BSC model) suppress needed updates and may artificially increase $M_S$, potentially shrinking the step size $\omega=\alpha(M_S+1)^{-\beta}$ prematurely and slowing exploration.
\end{itemize}

These effects are implementation-dependent, and a detailed quantitative analysis typically requires additional assumptions on the control landscape and on the distribution of update directions. Nevertheless, regardless of such details, the certification bottleneck established in {\bf Theorem~\ref{thm:noise_blowup}} remains: even an ideal search cannot overcome the exponential rarity of
long success runs once $qM_H$ is not small.

\begin{remark}[Implication for algorithm design]\label{rem:design_implication}
The analysis suggests a clear design rule for SSML with run-length halting: for a given classification-noise level $q$, one should not choose $M_H$ much larger than $\Theta(1/q)$ unless an
exponentially large shot budget is acceptable. Alternatively, one may modify the halting rule (e.g., windowed success rates, sequential likelihood-ratio tests, or majority-vote verification steps) so that certification uses redundancy in a way that is less sensitive to single-bit flips.
\end{remark}

%-------------------------------------------------------------------------------------------------------------------------------------------------------------------------------------------------------------------------------------
\section{Numerical simulations and empirical validation}\label{sec:numerics}
%-------------------------------------------------------------------------------------------------------------------------------------------------------------------------------------------------------------------------------------

This section validates the main stopping-time predictions of this work by numerical simulation. The goal is not to provide empirical evidence for the two structural claims that underlie our analysis:
\begin{itemize}
\item[1.] In the noiseless regime, SSML exhibits an essentially $O(N^{-1})$ accuracy law once the halting rule is viewed as an intrinsic sequential certificate.
\item[2.] Under classification noise with flip probability $q$, the run-length certificate induces a universal exponential blow-up in the halting time once $qM_H \gtrsim 1$.
\end{itemize}

%---------------------------------------------------------------------------------------------------------------------------------------------------------------------------------
\subsection{Simulation model and SSML implementation}\label{subsec:numerics_model}

{\em Unknown states and fiducial measurement.}---For a given Hilbert-space dimension $d$, we sample unknown pure states $\ket{\psi_\tau}$ from the Haar measure~\cite{Dankert2009,Harrow2009,Brandao2016}. We fix the fiducial state as $\ket{f}=\ket{0}$ in the computational basis and use the multi-outcome projective measurement $\{\ketbra{j}{j}\}_{j=0}^{d-1}$ after applying the current learner unitary $\hat{U}$. The ``success'' outcome corresponds to $j=0$ and any $j\ge 1$ is treated as ``failure'' (with a label that records which failure outcome occurred).

{\em High-dimensional SSML update rule.}---In $d>2$, we implement the standard SSML extension in which, upon a failure outcome $j=k\in\{1,\dots,d-1\}$, the feedback applies a random SU($2$) perturbation on the two-dimensional subspace $\mathrm{span}\{\ket{0}, \ket{k}\}$ while acting trivially on the orthogonal complement. Concretely, letting $R_{(0,k)}(\omega)\in \mathrm{SU}(2)$ be a random rotation of size $\omega$, we embed it into $\mathbb{C}^{d \times d}$ as a unitary $\hat{V}_{(0,k)}(\omega)$ and update
\begin{eqnarray}
\hat{U} \leftarrow \hat{V}_{(0,k)}(\omega)\,\hat{U},
\label{eq:num_update}
\end{eqnarray}
with $\omega = \alpha (M_S+1)^{-\beta}$ and $M_S \leftarrow 0$, while on success we keep $\hat{U}$ unchanged and increment $M_S\leftarrow M_S+1$. Here, we use the hyperparameters $\alpha=0.3$ and $\beta=1/2$ throughout.

{\em Stopping time and error metric.}---A run terminates at the stopping time $T=\inf\{n\ge 1:\ M_S^{(n)}=M_H\}$. The state estimate is $\ket{\psi_{\rm est}}=\hat{U}^\dagger\ket{0}$ and we evaluate the infidelity $\epsilon_T = 1-\abs{\braket{\psi_{\rm est}}{\psi_\tau}}^2$ at the stopping time.

%---------------------------------------------------------------------------------------------------------------------------------------------------------------------------------
\subsection{Noiseless validation: learning probability and $O(N^{-1})$ scaling}\label{subsec:numerics_noiseless}

{\em Learning probability.}---For fixed $M_H$ we estimate the learning probability $P(N)=\Pr[T \le N]$ by repeating independent runs. Fig.~\ref{fig:num_learning_prob} shows that $P(N)$ is well approximated by an exponential CDF of the form
\begin{eqnarray}
P(N) \approx 1-\exp(-N/N_c),
\label{eq:num_PN_fit}
\end{eqnarray}
with a characteristic scale $N_c$ comparable to $\mathbb{E}[T]$. This corroborates the stopping-time viewpoint that SSML completion is governed by a well-defined hitting-time distribution, not by rare pathological events.

\begin{figure}[t]
\centering
\includegraphics[width=1.00\linewidth]{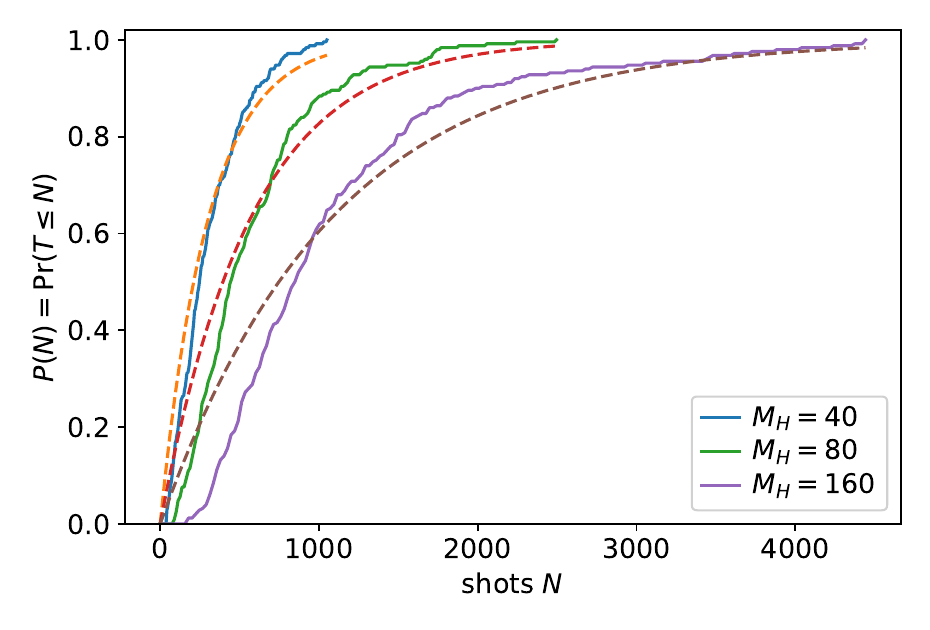}
\caption{The empirical learning probability $P(N)=\Pr[T \le N]$ for $d=2$ and several halting thresholds $M_H$ (solid), together with the fit $1-\exp(-N/N_c)$ (dashed).}
\label{fig:num_learning_prob}
\end{figure}

{\em Accuracy scaling and dimension dependence.}---To validate the predicted near-optimal accuracy, we sweep $M_H$ and record the pair $(\mathbb{E}[T],\mathbb{E}[\epsilon_T])$. Fig.~\ref{fig:num_infid_scaling} shows an approximately linear relation $\mathbb{E}[\epsilon_T] \propto 1/\mathbb{E}[T]$ on a log-log scale, consistent with the $O(N^{-1})$ law. The separation of roles emphasized in our analysis is also visible: the dimension $d$ mainly affects the prefactor (through the search geometry), while the certificate mechanism yields the universal $1/N$ dependence once halting is feasible.

\begin{figure}[t]
\centering
\includegraphics[width=1.00\linewidth]{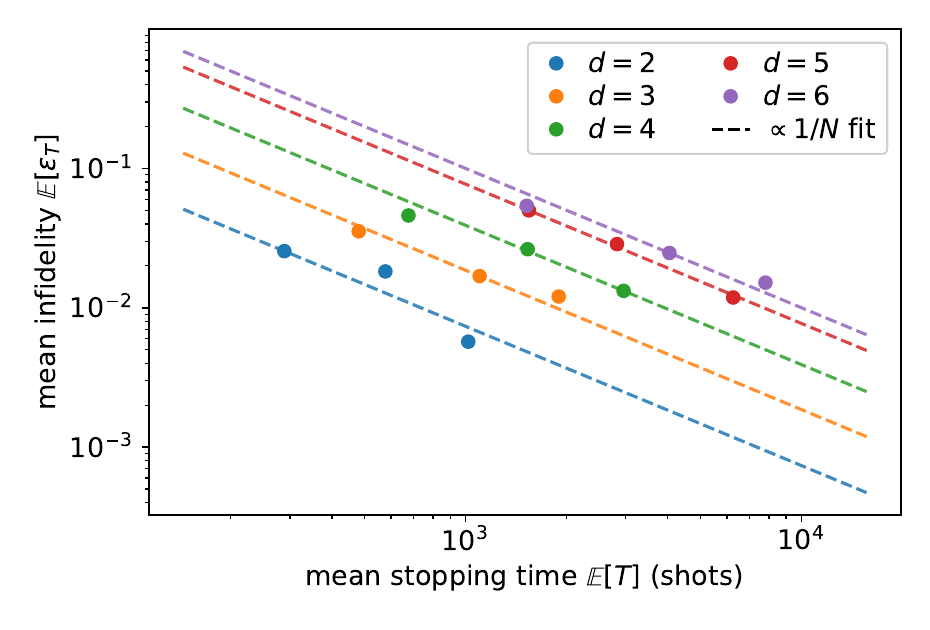}
\caption{Noiseless SSML accuracy scaling.  Each marker corresponds to a halting threshold $M_H\in\{40,80,160\}$, and the dashed line indicates a reference $1/N$ slope.  Increasing $d$ shifts the curve mainly through a dimension-dependent prefactor, consistent with the factor $\mathsf{K}(d)$ in our bounds.}
\label{fig:num_infid_scaling}
\end{figure}

%---------------------------------------------------------------------------------------------------------------------------------------------------------------------------------
\subsection{Classification-noise validation: feasibility threshold and exponential blow-up}\label{subsec:numerics_noise}

Our analysis in Sec.~\ref{sec:noise_analysis} shows that classification noise induces a universal certification bottleneck. To validate this effect numerically without conflating it with search degradation, we simulate the best-case certification process: an i.i.d. Bernoulli record with success probability $1-q$, corresponding to perfect control ($F=1$) but noisy labels. In this setting, the SSML halting time is exactly the waiting time $X_{M_H}(1-q)$ for $M_H$ consecutive recorded successes.

Fig.~\ref{fig:num_noise_collapse} plots the scaled mean $q \mathbb{E}[X_{M_H}(1-q)]$ against the dimensionless noise-load $qM_H$. The simulation collapses across different $q$, agreeing with the theoretical prediction $q \mathbb{E}[X_{M_H}(1-q)] \approx e^{qM_H}-1$. This numerically confirms the sharp operational threshold: once $qM_H \gtrsim 1$, the expected halting time grows exponentially, rendering completion effectively unattainable for any fixed shot budget~\cite{Lee2021,Drekic2021}.

\begin{figure}[t]
\centering
\includegraphics[width=1.00\linewidth]{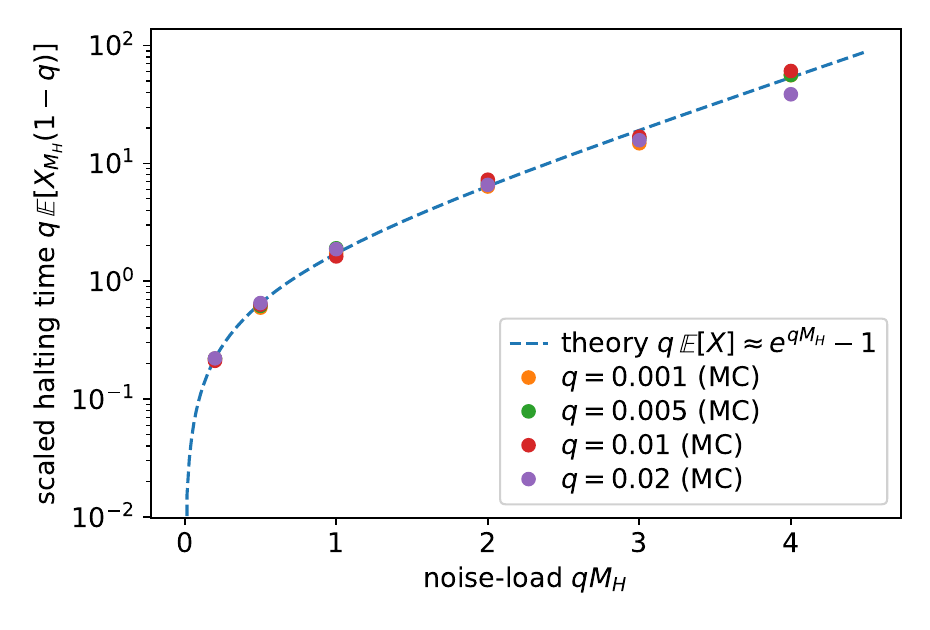}
\caption{The universal noise-induced blow-up of the run-length certificate. Markers: Monte-Carlo estimates of $q\mathbb{E}[X_{M_H}(1-q)]$ for several noise rates $q$. Dashed curve: theoretical prediction $e^{qM_H}-1$ (valid for small $q$ and accurate as a scaling law more generally). The collapse highlights $qM_H$ as the controlling dimensionless parameter and illustrates the sharp feasibility threshold at $qM_H=O(1)$.}
\label{fig:num_noise_collapse}
\end{figure}

%-------------------------------------------------------------------------------------------------------------------------------------------------------------------------------------------------------------------------------------
\section{Information-theoretic and physical interpretation}\label{sec:info_interpretation}
%-------------------------------------------------------------------------------------------------------------------------------------------------------------------------------------------------------------------------------------

In this section, we step back and interpret the results as statements about how information is acquired from quantum systems under strict physical constraints. Our goal is not merely to ``explain'' the obtained scalings, but to articulate why SSML can be viewed as a \emph{minimal model of adaptive information acquisition}, where the halting condition plays the role of an intrinsic sequential certification test and where imperfections in the classical feedback channel induce a genuine information-theoretic bottleneck.

%---------------------------------------------------------------------------------------------------------------------------------------------------------------------------------
\subsection{SSML as a minimal model of adaptive information acquisition}\label{subsec:info_minimal_model}

At a high level, SSML is a sequential interaction protocol between a quantum source and a classical controller. At each step $n$, the controller chooses an action $\mathbf{p}^{(n)}$ (equivalently, a unitary
$\hat{U}(\mathbf{p}^{(n)})$), observes a single binary outcome $m_{n}\in\{s,u\}$ from the two-outcome test $\{\hat{M}_{f}, \hat{M}_{f^{\perp}}\}$, and updates $\mathbf{p}$ according to a rule that uses only this bit and a minimal internal memory $M_{S}^{(n)}$. This creates a particularly sparse information flow:
\begin{center}
$1$ unknown copy $\rightarrow$ $1$ classical bit $\rightarrow$ $1$ update decision.
\end{center}

In comparison, the standard tomography extracts a multi-outcome record and then performs extensive post-processing~\cite{Hradil1997,James2001}, whereas many adaptive schemes still rely on nontrivial intermediate estimation (e.g., Bayesian filtering~\cite{Huszar2012,Granade2016}). SSML, in contrast, is intentionally austere: it uses (i) a one-bit measurement record, (ii) a tiny memory ($M_{S}$), and (iii) a purely local, failure-triggered random perturbation. From an information-science viewpoint, this austerity is not a limitation but the point: it isolates the minimal ingredients required for near-optimal learning under repeated-copy access~\cite{Massar1995,Gill2000,Bagan2006,Lee2021}.

A useful perspective is to regard the measurement outcome as the output of an adaptive binary channel whose parameter is the success probability $F(\mathbf{p})$ in Eq.~(\ref{eq:ssml_F_eps}). SSML repeatedly queries this channel at adaptively chosen $\mathbf{p}$, attempting to steer $F(\mathbf{p})$ toward unity. In this analogy, the unknown quantum state is not a parameter-observed directly; it is latent and enters only through the mapping
\begin{eqnarray}
\ket{\psi_\tau}\ \mapsto\ F(\mathbf{p})=\abs{\bra{f}\hat{U}(\mathbf{p})\ket{\psi_\tau}}^{2},
\label{eq:info_latent_mapping}
\end{eqnarray}
and quantum mechanics constrains which such channels are physically realizable.

\begin{remark}[Minimality in practice]\label{rem:info_minimality}
SSML has three strong ``minimality'' features that are often absent simultaneously in adaptive estimation: (i) no quantum memory or collective measurement across copies, (ii) no full intermediate estimator of $\ket{\psi_\tau}$ is maintained during learning, and (iii) the algorithm's internal state is effectively the pair $(\mathbf{p},M_S)$. This makes SSML a clean testbed for separating what is genuinely quantum-limited (copy complexity, no-cloning) from what is algorithmic (search geometry) and what is classical (feedback reliability).
\end{remark}

%---------------------------------------------------------------------------------------------------------------------------------------------------------------------------------
\subsection{Halting as sequential certification: an evidence budget}\label{subsec:info_halting_certificate}

The central conceptual step of this work is to interpret the halting condition $M_{S}=M_H$ as an intrinsic sequential certificate. This is not an optional ``after-the-fact'' verification; it is the operational definition of completion.

In the noiseless case, once the algorithm is near a good control $\mathbf{p}$, success events freeze the control, and the final run of $M_H$ successes is generated under a fixed $\mathbf{p}_{\rm est}$. Conditional on $\mathbf{p}_{\rm est}$, the halting event is therefore a statement about $F(\mathbf{p}_{\rm est})$. {\bf Lemma~\ref{lem:certificate_strength}} made this explicit by showing
\begin{eqnarray}
\mathbb{P}\bigl(\text{$M_H$ consecutive successes} \mid \mathbf{p}\bigr) = F(\mathbf{p})^{M_H}.
\label{eq:info_cert_prob}
\end{eqnarray}
Thus, for any target infidelity scale $\epsilon_{0}$, the certificate error exponent is $M_H$:
\begin{eqnarray}
\epsilon(\mathbf{p})\ge \epsilon_{0} \ \Longrightarrow \ \mathbb{P}\left(\text{halt at $\mathbf{p}$} \mid \mathbf{p}\right) \le (1-\epsilon_{0})^{M_H}.
\label{eq:info_cert_exponent}
\end{eqnarray}
In other words, $M_H$ is an evidence budget measured in consecutive successful ``yes'' answers. This is why the stopping-time viewpoint is natural: the protocol continues to spend copies until the accumulated evidence (a success run) reaches a prescribed threshold.

This viewpoint also clarifies why the search--certification decomposition of Sec.~\ref{subsec:noiseless_decomposition} and Sec.~\ref{sec:noiseless_bounds} is more than a technical convenience. The search stage is tasked with bringing the system into a regime where the certificate is operationally meaningful; the certification stage is then governed by universal run statistics ({\bf Theorem~\ref{thm:runs_mean_tail}}) and is essentially independent of $d$ once a lower bound on $F$ is available. In particular, the ``linear-versus-exponential'' dichotomy of the certification cost (Eq.~(\ref{eq:cert_scaling_recall_V})) provides a crisp design principle: $\epsilon(\mathbf{p}_{\rm est})M_H = O(1)$ is the regime where halting is both meaningful and feasible.

\begin{remark}[Sequential testing without explicit likelihood ratios]\label{rem:info_SPRT}
The halting rule is structurally analogous to a sequential test: it uses a streaming record and stops once the evidence crosses a threshold.  Unlike Wald-type likelihood-ratio tests, SSML uses a fixed summary statistic ($M_S$) whose update requires negligible computation.  The price of this simplicity is that the certificate is tied to run statistics and therefore becomes highly sensitive to label noise (Sec.~\ref{sec:noise_analysis}).
\end{remark}

%---------------------------------------------------------------------------------------------------------------------------------------------------------------------------------
\subsection{Calibration against quantum limits: why $O(N^{-1})$ is natural}\label{subsec:info_quantum_limits}

A recurring theme in the SSML literature is that the observed learning accuracy is (nearly) optimal, with infidelity scaling close to $O(N^{-1})$ as a function of the number $N$ of consumed copies~\cite{Lee2018,Lee2021}. This is not an accident: for pure-state estimation, $O(N^{-1})$ is the maximum statistical accuracy that can be expected on general quantum-estimation grounds, and it can be understood as a manifestation of the no-cloning principle~\cite{Massar1995,Gill2000,Bagan2006,Wootters1982,Dieks1982}. In particular, it is emphasized in the original SSML proposal that $O(N^{-1})$ represents the ultimate accuracy limit imposed by quantum-state estimation theory, or equivalently by no cloning~\footnote{The statement is standard: if one could estimate an arbitrary unknown pure state with infidelity scaling asymptotically faster than $1/N$, then one could turn $N$ copies into a classical description accurate enough to reproduce (clone) the state at a rate incompatible with universal quantum cloning limits (See the discussion in Sec.~I of Ref.~\cite{Lee2018}; see also Refs.~\cite{Buvzek1996,Werner1998}).}.

Our results provide a complementary and arguably more operational explanation of why SSML sits naturally at this limit.

\smallskip
\noindent (i) {\em Certification inherently costs $O(1/\epsilon)$.}---In the high-fidelity regime, the monitored run length is geometric and satisfies $\mathbb{E}[M_S\mid \epsilon]\approx \epsilon^{-1}$ (Sec.~\ref{subsec:noiseless_MS_proxy}). Thus, to \emph{see} evidence consistent with infidelity $\epsilon$, SSML must spend on the order of $1/\epsilon$ copies simply to accumulate a long run of successes. This is a sequential version of a familiar information constraint: very small error probabilities can only be resolved after many trials.

\smallskip
\noindent (ii) {\em Search converts geometry (dimension) into a prefactor.}---The search stage determines how quickly the learner reaches a regime where $\epsilon=O(1/M_H)$. This stage carries the dependence on $d$ (through the control landscape and effective parameter dimension). Once this dependence is extracted into a factor $\mathsf{K}(d)$, the remaining scaling is essentially the same as in a one-parameter Bernoulli certification problem~\cite{Drekic2021}.

\smallskip
\noindent (iii) {\em Eliminating $M_H$ yields $O(N^{-1})$ with a transparent meaning.}---In Sec.~\ref{subsec:noiseless_combine}, we showed how the certificate scale $\epsilon_{\rm cert}(M_H,\delta)\approx \ln(1/\delta)/M_H$ and the stopping-time bounds together imply
\begin{eqnarray}
\epsilon(\mathbf{p}_{\rm est}) = O\Big(\mathsf{K}(d)\frac{\ln(1/\delta)}{N}\Big),
\label{eq:info_eps_vs_N_recall}
\end{eqnarray}
up to modest logarithmic factors and model-dependent prefactors. From the present interpretative standpoint, Eq.~(\ref{eq:info_eps_vs_N_recall}) says: SSML reaches the quantum-limited $1/N$ accuracy because its halting rule enforces a certificate whose evidence budget scales linearly with the available copies.

The above description (i)--(iii) helps justify the viewpoint that SSML is not merely a convenient heuristic. Rather, it is a concrete model in which a one-bit-per-copy measurement record, a minimal memory variable, and a stopping-time certificate conspire to achieve the natural quantum accuracy limit without collective measurements or computationally heavy estimators.

%---------------------------------------------------------------------------------------------------------------------------------------------------------------------------------
\subsection{Noise threshold as an information bottleneck}\label{subsec:info_noise_bottleneck}

The most striking qualitative phenomenon uncovered by the stopping-time formulation is the sharp noise-induced blow-up of the halting time (Sec.~\ref{sec:noise_analysis}). At first sight, one might expect classification noise to merely ``slow down'' learning by confusing the update rule. Our analysis shows something stronger and more universal: the classification noise can render the very notion of completion operationally unattainable.

The reason is information-theoretic and stems from the certification semantics. Under label noise with flip probability $q$, the observable success probability is capped by $\widetilde{F}_{\max}=1-q$, even if the underlying physical control can achieve $F(\mathbf{p})=1$. Consequently, the probability of observing $M_H$ recorded successes in a row is at most $(1-q)^{M_H}$, and the waiting time is at least of order $(1-q)^{-M_H}$. This yields the universal lower bound $\mathbb{E}[T]\ \gtrsim\ \frac{1}{q}\,e^{qM_H}~(qM_H \gg 1)$ ({\bf Theorem~\ref{thm:noise_blowup}}), together with analogous exponential behavior for high-confidence halting times.

A particularly transparent way to read the threshold is the following. During the final certificate segment, which consists of $M_H$ recorded ``success'' labels, the expected number of corrupted labels is $qM_H$. Thus,
\begin{widetext}
\begin{eqnarray}
qM_H \lesssim 1 \quad\Longleftrightarrow\quad \text{``one expects fewer than one corruption per certificate.''}
\label{eq:info_qMH_interpretation}
\end{eqnarray}
\end{widetext}
When $qM_H \gtrsim 1$, clean success runs become exponentially rare, and halting becomes exponentially unlikely within any feasible shot budget. This is an operational threshold: it does not mean the system is unlearnable in principle, but that the run-length certificate ceases to be a practical completion criterion under the given feedback reliability.

\begin{remark}[Connection to experimentally observed SNR ceilings]\label{rem:info_SNR_ceiling}
In single-shot implementations, the false-negative events (dark counts, finite extinction ratios, etc.) directly interfere with long runs of successes. This leads to an empirical ceiling on measurable run lengths of order $\mathrm{SNR}\sim 1/q$, beyond which $M_S$ cannot reliably grow before being reset by noise. Our exponential blow-up bound provides a stopping-time explanation of this observation: the long runs are suppressed at precisely the scale where $qM_H$ ceases to be small.
\end{remark}

Finally, we emphasize a conceptual asymmetry revealed by the analysis. In the noiseless setting, increasing $M_H$ strengthens the certificate and (through the search--certification decomposition) reduces infidelity in a controlled manner. Under the classification noise, increasing $M_H$ beyond $\Theta(1/q)$ no longer strengthens the certificate in an operationally meaningful way, because the cost of witnessing the corresponding run explodes. Thus, the noise imposes a certificate ceiling that is not a property of the quantum system alone, but of the full quantum--classical loop.

%---------------------------------------------------------------------------------------------------------------------------------------------------------------------------------
\subsection{Physical principles and broader implications}\label{subsec:info_physical_principles}

We now connect the stopping-time picture to basic physical principles. The point is not to invoke foundational constraints, but to show that the structural features identified in Sec.~\ref{sec:noiseless_bounds} and Sec.~\ref{sec:noise_analysis} admit a crisp physical interpretation once SSML is viewed as an adaptive information-acquisition loop with an intrinsic certificate.

\medskip
\noindent \textbf{No-cloning as a resource law.} A universal cloning machine would convert $N$ copies of an unknown state into $N'$ copies with fidelity exceeding what is permitted by quantum mechanics~\cite{Dieks1982,Buvzek1996,Werner1998}. Equivalently, no-cloning constrains how efficiently one can turn finite-copy access into a classical description that is accurate enough to reproduce (prepare) the state on demand.

SSML makes this conversion explicit. At completion, the learner outputs a classical controller state $\mathbf{p}_{\rm est}$, which can be used to reproduce an approximation of the unknown target as $\hat{U}(\mathbf{p}_{\rm est})^{\dagger}\ket{f}$. The stopping-time bounds then read naturally as a \emph{resource law}: to certify infidelity $\epsilon$ using a one-bit-per-copy record, SSML must accumulate on the order of $1/\epsilon$ copies, because (i) the expected run length at infidelity $\epsilon$ is $\Theta(1/\epsilon)$ [Eq.~(\ref{eq:geom_mean})], and (ii) the halting rule requires a macroscopic amount of such evidence ({\bf Theorem~\ref{thm:runs_mean_tail}}). If one could violate this law---e.g., guarantee infidelity $o(1/N)$ with only $N$ copies---then the resulting classical description would be accurate enough to synthesize additional high-fidelity copies at a rate that would contradict known cloning/estimation limits. From this perspective, the familiar $O(N^{-1})$ frontier for pure-state estimation is not merely an asymptotic Cram\'er--Rao statement; it is enforced operationally by the evidence budget implicit in sequential certification.

\medskip
\noindent \textbf{No-signaling and causal adaptivity.} SSML is a strictly causal, copy-by-copy protocol: the choice of control $\mathbf{p}^{(n)}$ is measurable with respect to the past filtration $\mathcal{F}_{n-1}$, and the protocol never requires quantum memory across copies or collective measurements. This architecture respects no-signaling in the strongest operational sense~\cite{NielsenChuang2000,Dieks1982}: there is no step at which a spacelike separated degree of freedom could be influenced by future measurement choices, and all adaptivity is implemented by classical feedforward.

The stopping-time formulation makes this causality explicit. First, the halting time $T$ is a stopping time adapted to $\{\mathcal{F}_{n}\}$, which precisely formalizes the idea that ``the decision to stop uses only past data.'' Second, the decomposition theorem ({\bf Theorem~\ref{thm:decomposition_bound}}) shows that once a lower bound on the instantaneous success probability is achieved, the remaining completion time is controlled by a universal Bernoulli run process. This universality is itself a consequence of causality: the protocol cannot access future outcomes to ``force'' a long run, and therefore must pay the genuine statistical cost of observing it. In this way, SSML provides a clean example in which causal implementability is not a background assumption but a mathematical ingredient that shapes resource scaling.

\medskip
\noindent \textbf{Implication: SSML as a canonical benchmark model.} Taken together, the search--certification decomposition, the calibration to the $1/N$ information limit, and the noise-feasibility threshold suggest that SSML is more than a specialized algorithm for a single estimation task. It is a canonical benchmark model for studying adaptive information acquisition under simultaneous constraints:
\begin{itemize}
\item[(i)] {\em Quantum constraints:} each copy yields only Born-rule-limited information, and no-cloning prevents one from amplifying this information into arbitrarily many faithful replicas.
\item[(ii)] {\em Architectural constraints:} SSML operates with one-bit feedback and minimal memory, making it a realistic abstraction of hardware-limited closed-loop control.
\item[(iii)] {\em Statistical constraints:} completion is defined by a sequential certificate, so sample usage is governed by stopping-time laws rather than fixed-budget accuracy formulas.
\end{itemize}

In this benchmark view, the most striking universality is the certification bottleneck under classification noise: even an idealized search cannot evade the exponential rarity of strict run-length events once $qM_H \gtrsim 1$ ({\bf Theorem~\ref{thm:noise_blowup}}). This isolates, in a mathematically clean way, how an imperfect classical record can become the dominant information-theoretic bottleneck of an otherwise quantum-limited learning loop.

Finally, the benchmark perspective also clarifies what it means to ``improve'' SSML. Many variants in the literature primarily engineer the search stage (e.g., tailoring the update distribution to avoid vanishing improvement probability in large parameter spaces~\cite{Spall1992}). Our analysis highlights a complementary axis: redesigning the certificate (the halting rule). Alternative stopping rules, e.g., likelihood-based sequential tests~\cite{Siegmund2013}, can preserve the minimal-feedback spirit while decoupling certificate strength from the fragility of strict run statistics. From the stopping-time standpoint, such variants correspond to changing the stopping rule and therefore changing the stopping-time law, providing a principled route to noise-tolerant adaptive learning.

%-------------------------------------------------------------------------------------------------------------------------------------------------------------------------------------------------------------------------------------
\section{Discussion}\label{sec:discussion}
%-------------------------------------------------------------------------------------------------------------------------------------------------------------------------------------------------------------------------------------

We developed a stopping-time sample-complexity theory for single-shot measurement learning (SSML), motivated by a simple but often underemphasized observation: SSML does not naturally fit the standard fixed-budget paradigm of quantum state estimation. The protocol defines completion through a run-length halting rule, so the natural resource is the random number of consumed copies until halting. By elevating the halting rule to an intrinsic sequential certificate, we reframed SSML as a minimal model of adaptive information acquisition under quantum and architectural constraints and obtained scaling laws that are both operational and conceptually interpretable.

%---------------------------------------------------------------------------------------------------------------------------------------------------------------------------------
\subsection{Summary of contributions}\label{subsec:conclusion_summary}

Our results can be summarized in three interconnected statements.

\medskip
\noindent (i) {\em Halting is certification, not bookkeeping.}---Because SSML freezes the control on each success ({\bf Remark~\ref{rem:ssml_freezing}}), the terminal run of $M_H$ consecutive successes is generated under a fixed control $\mathbf{p}_{\rm est}$. This makes the halting event an intrinsic sequential certificate of high success probability and hence of low infidelity ({\bf Lemma~\ref{lem:certificate_strength}}).

\medskip
\noindent (ii) {\em The stopping time decomposes into search and certification.}---The stopping-time formulation isolates two qualitatively distinct costs: (a) {\bf search}, the landscape- and dimension-dependent effort required to bring the success probability close to unity, and (b) {\bf certification}, the dimension-independent overhead of observing $M_H$ consecutive successes once this regime
is reached. This separation is made explicit by the universal decomposition bound ({\bf Theorem~\ref{thm:decomposition_bound}}) and by the run-statistics law ({\bf Theorem~\ref{thm:runs_mean_tail}}).

\medskip
\noindent (iii) {\em Classification noise induces an operational threshold.}---Under label flips with probability $q$, the observable success probability is capped away from one even for ideal control. Because completion is defined by a strict run-length event, the expected stopping time must exhibit an exponential blow-up once $qM_H$ is not small ({\bf Theorem~\ref{thm:noise_blowup}}). This yields a sharp operational feasibility criterion: beyond the regime $qM_H\lesssim 1$, ``learning completion'' is effectively unattainable within any realistic shot budget.

Taken together, these statements support the central thesis of this work: SSML is not merely a convenient estimation algorithm, but a physically grounded model in which sequential certification, adaptive search, and information constraints are inseparable and can be analyzed in a unified stopping-time framework.

%---------------------------------------------------------------------------------------------------------------------------------------------------------------------------------
\subsection{Operational lessons and protocol-design implications}\label{subsec:conclusion_design}

The stopping-time bounds translate into direct guidance for using and modifying SSML in practice. In the noiseless setting, increasing $M_H$ strengthens the intrinsic certificate: the probability that a fixed
low-fidelity control produces $M_H$ consecutive successes decays exponentially in $M_H$ ({\bf Lemma~\ref{lem:certificate_strength}}), yielding the certificate scale $\epsilon_{\rm cert}(M_H,\delta)\approx \ln(1/\delta)/M_H$ in Eq.~(\ref{eq:eps_cert}). However, $M_H$ also determines the certification overhead through run statistics, and therefore must be chosen with feasibility in mind. {\bf Remark~\ref{rem:linear_vs_exp_cert}} identifies the relevant regime: certification remains efficient only when the underlying infidelity is already within $O(1/M_H)$ of zero.

Under the classification noise, this trade-off becomes sharper. The observable success probability saturates below unity, which makes long success runs exponentially unlikely. As a consequence, one cannot arbitrarily ``push'' SSML to higher apparent confidence by increasing $M_H$: once $qM_H\gtrsim 1$ the expected halting time grows essentially as $\exp(qM_H)$ ({\bf Theorem~\ref{thm:noise_blowup}}). A pragmatic design rule is therefore to keep the halting threshold below the noise-limited run scale,
\begin{eqnarray}
M_H \lesssim \frac{c}{q},
\label{eq:concl_MH_rule}
\end{eqnarray}
for a modest constant $c=O(1)$ determined by the acceptable overhead.

Our analysis also suggests a principled direction for protocol engineering in noisy regimes. The bottleneck is often not the search dynamics itself but the fragility of the certificate---a strict run-length statistic on a noisy one-bit record~\cite{Cover1999,Drekic2021}. This motivates a ``hybrid'' design philosophy: retain the minimal-memory, one-bit SSML search loop, but replace the halting rule by a more noise-tolerant sequential test~\cite{Siegmund2013}. In the stopping-time language, such modifications correspond to changing the stopping rule, and hence, changing the stopping-time law; the analysis developed here applies directly to these variants.

Finally, the decomposition highlights where dimension dependence lives. While certification is governed by universal run statistics, the search stage can deteriorate in high-dimensional parameterizations because random perturbations become less likely to align with improving directions. This explains why practical high-dimensional SSML variants engineer the update distribution to avoid vanishing improvement probability, effectively compressing the search into lower-dimensional subspaces.

%---------------------------------------------------------------------------------------------------------------------------------------------------------------------------------
\subsection{Physical and information-theoretic meaning}\label{subsec:conclusion_info}

Beyond algorithm design, the stopping-time viewpoint provides a narrative that connects SSML to foundational constraints. First, SSML converts finite-copy access into a classical controller state that can be used to reproduce an approximation of the unknown quantum state. The resulting $1/\epsilon$ resource law for certification (Sec.~\ref{subsec:noiseless_certification}) is consistent with the information limits of pure-state estimation and, equivalently, with no-cloning constraints argued in the original SSML proposal~\cite{Lee2018,Massar1995,WoottersZurek1982}. Our contribution is to show how this limit emerges operationally from the sequential evidence accumulation: the protocol cannot certify infidelity $\epsilon$ without spending $\Theta(1/\epsilon)$ copies because the certificate statistic itself requires that amount of data.

Second, the noise threshold illustrates how physical imperfections translate into information-theoretic bottlenecks. The classification noise does not merely introduce a small bias; it limits the reliability of the one-bit feedback channel on which SSML relies to both update and certify completion. The condition $qM_H \gtrsim 1$ can be read as the statement that during the would-be certificate run one expects $O(1)$ corrupted labels, rendering the strict run-length event exponentially rare. In this sense, the blow-up is certification-limited and landscape independent: it persists even if the search stage is replaced by an ideal oracle.

Finally, SSML is inherently causal and implementable: each decision is based only on past information and each copy is measured locally and then discarded. The stopping-time formulation makes this causality explicit through the filtration-adapted definition of $T$ and through coupling arguments that compare SSML to idealized Bernoulli run processes. This provides a concrete example in which causal implementability is not only physically required but mathematically shapes resource scaling~\cite{Durrett2019,Shaked2007}.

%---------------------------------------------------------------------------------------------------------------------------------------------------------------------------------
\subsection{Limitations and future directions}\label{subsec:conclusion_future}

Several extensions follow naturally from the present analysis. On the noise side, real experiments exhibit structured imperfections, such as correlated mislabeling, drift, and systematic miscalibration (see Refs.~\cite{Lee2021}). Replacing the i.i.d. flip model by Markovian or bursty noise models, it would be interesting to determine whether the threshold condition $qM_H$ is replaced by an effective ``noise load'' accumulated over the certificate window. On the modeling side, SSML is fundamentally aligned with pure-state learning via unitary inversion; mixed targets or model mismatch can prevent halting or induce a noise-limited ceiling. A stopping-time theory for mixed-state learning likely requires a modified objective or a certificate statistic that does not rely on approaching unit fidelity.

A particularly important direction is to classify alternative halting rules as sequential tests. From our perspective, the halting condition is the central information-theoretic object in SSML. Different stopping rules (run-length, windowed rates, likelihood ratios, verification bursts) correspond to different stopping-time laws and therefore to different robustness and sample-complexity trade-offs. Developing a general ``certificate complexity'' theory for minimal-memory adaptive quantum learning remains an interesting interface between quantum estimation and sequential analysis~\cite{Siegmund2013}.

Finally, our results emphasize scaling laws and universal mechanisms. Sharper constants and genuinely tight bounds will require a more detailed description of the control landscape and of the update distribution, especially in high-dimensional parametrizations. Such refinements could connect SSML more directly to modern analyses of stochastic optimization~\cite{Spall2002}, bandit feedback~\cite{Bubeck2012}, and reinforcement learning under partial observability~\cite{Jae2025}.

%---------------------------------------------------------------------------------------------------------------------------------------------------------------------------------
\subsection{Conclusion}\label{subsec:conclusion_final}

In conclusion, SSML admits a natural reformulation as a stopping-time learning problem in which halting is an intrinsic sequential certificate rather than an external stopping criterion. This viewpoint yields a clean separation of sample costs into search and certification and exposes a universal trade-off between target accuracy (set by $M_H$), Hilbert-space dimension $d$, and the reliability of the one-bit feedback record. Most notably, classification noise induces an exponential blow-up of the halting time once $qM_H \gtrsim 1$, constituting a sharp operational threshold that is invisible to fixed-budget analyses. Beyond quantifying performance, these results support the broader claim that SSML is a conceptually important model for adaptive information acquisition in quantum systems, where physical principles and information constraints can be studied through an experimentally implementable minimal-feedback architecture.

%=========================================================================================================================================
\section*{Acknowledgement}
This work was supported by the Ministry of Science, ICT and Future Planning (MSIP) by the National Research Foundation of Korea (RS-2024-00432214, RS-2025-03532992, and RS-2025-18362970) and the Institute of Information and Communications Technology Planning and Evaluation grant funded by the Korean government (RS-2019-II190003, ``Research and Development of Core Technologies for Programming, Running, Implementing and Validating of Fault-Tolerant Quantum Computing System''), the Korean ARPA-H Project through the Korea Health Industry Development Institute (KHIDI), funded by the Ministry of Health \& Welfare, Republic of Korea (RS-2025-25456722), and the Ministry of Trade, Industry and Resources (MOTIR), Korea, under the project ``Industrial Technology Infrastructure Program'' (RS-2024-00466693). We acknowledge the Yonsei University Quantum Computing Project Group for providing support and access to the Quantum System One (Eagle Processor), which is operated at Yonsei University.

\end{document}